\documentclass[10pt]{article}

\usepackage[english]{babel}
\usepackage[T1]{fontenc}

\usepackage[lmargin=1in,rmargin=1in,tmargin=1in,bmargin=1in]{geometry}

\usepackage[affil-it]{authblk}  
\usepackage{amsmath}
\usepackage{graphicx}
\usepackage{amsfonts}
\usepackage{amsthm}
\usepackage{mathtools}
\usepackage{subcaption} 
\usepackage{comment}
\usepackage{float}
\usepackage[colorinlistoftodos]{todonotes}
\usepackage[%
bookmarks=true,
colorlinks=true,
linkcolor=black,
urlcolor=black,
citecolor=black,
plainpages=false,
pdfpagelabels,
final]{hyperref}

\usepackage[capitalise,nameinlink]{cleveref}

\usepackage{amssymb}

\numberwithin{equation}{section}

\theoremstyle{plain}
\newtheorem{theorem}{Theorem}
\newtheorem{proposition}{Proposition}[section]
\newtheorem{lemma}[proposition]{Lemma}
\newtheorem*{theorem*}{Theorem}
\newtheorem{corollary}{Corollary}[section]
\newtheorem*{conjecture*}{Conjecture}

\newtheorem*{openpb*}{Open problem}

\theoremstyle{definition}
\newtheorem{definition}[proposition]{Definition}
\newtheorem{remark}[proposition]{Remark}

\usepackage{csquotes}
\usepackage[backend=biber,style=alphabetic,giveninits=true,doi=false,isbn=false,url=false,sorting=nyt,maxbibnames=99,date=year]{biblatex}
\renewbibmacro{in:}{}
\numberwithin{equation}{section}
\title{Discretely self-similar exterior-naked singularities \\ for the Einstein-scalar field system}

\newcommand{\Rext}{\mathcal{R}_{\mathrm{Ext}}}

\author[1]{Serban~Cicortas\thanks{cicortas@princeton.edu}}
\author[2]{Christoph~Kehle\thanks{kehle@mit.edu}}

\affil[1]{\small  Department of Mathematics, Princeton University,
	
Fine Hall,	  Washington Road, Princeton NJ 08544, United States of America \vskip.1pc \ } 

\affil[2]{\small  Department of Mathematics, Massachusetts Institute of Technology, 
	
Building 2, 77 Massachusetts Avenue, Cambridge, MA 02139,  United States of America 
}
\date{12 December 2024}
\begin{document}
\maketitle
\begin{abstract} 
The problem of constructing naked singularities in general relativity can be naturally divided into two parts: 
\begin{enumerate}
    \item[(i)] the construction of the region exterior to the past light cone of the singularity, extending all the way to (an incomplete) future null infinity and yielding the nakedness property (what we will call \textit{exterior-naked singularity regions})
    \item[(ii)] attaching an interior fill-in that ensures that the singularity arises from regular initial data.
\end{enumerate}
This problem has been resolved for the spherically symmetric Einstein-scalar field system by Christodoulou \cite{C94}, but his construction, based on a continuously self-similar ansatz, requires that both the exterior and the interior regions are mildly irregular on the past cone of the singularity. On the other hand, numerical works suggest that there exist naked singularity spacetimes with discrete self-similarity arising from smooth initial data. In this paper, we revisit part (i) of the problem and we construct exterior-naked singularity regions with discretely self-similar profiles which are smooth on the past cone of the singularity. We show that the scalar field remains uniformly bounded, but the singularity is characterized by the infinite oscillations of the scalar field and the mass aspect ratio. (Our examples require however that the mass aspect ratio is uniformly small, and thus the solutions are distinct from the exterior regions of the numerical examples.) It remains an open problem to smoothly attach interior fill-ins as in (ii) to our solutions, which would yield a new construction of naked singularity spacetimes, now arising from smooth initial data.

\end{abstract}

\section{Introduction}
The Einstein-scalar field system 
\begin{equation}\label{Einstein Scalar Field}
    \begin{dcases}
    Ric_{\mu\nu}=2\partial_{\mu}\phi\partial_{\nu}\phi \\  
    \square_g\phi=0
    \end{dcases}
\end{equation}
for a spherically symmetric spacetime $\big(M^{3+1},g\big)$ and a scalar field $\phi$  is a well-established model for studying the physics of gravitational collapse \cite{Chr71}. The spacetime $\big(M^{3+1},g\big)$ admits double null coordinates $(u,v,\vartheta,\varphi)$ such that the metric takes the form
\begin{equation}\label{double null metric}
    g=-\Omega^2 (u,v) dudv+r^2(u,v)d\sigma_{S^2},
\end{equation}
where $d\sigma_{S^2} \doteq d \vartheta^2 + \sin^2 \vartheta d \varphi^2$ is the standard metric on the sphere $S^2$ and the function $r(u,v)$ is the area radius of the orbits of the $SO(3)$ action.

For asymptotically flat spherically symmetric initial data for \eqref{Einstein Scalar Field}, solutions arising from small data disperse to Minkowski space, while there exist large-data solutions that contain a black hole region. A third class of solutions that play a special role are \emph{naked singularity spacetimes}. Informally, these spacetimes contain a singularity that is visible to far away observers and not hidden behind the event horizon of a black hole. An alternative way to think about naked singularities is that they cause the \textit{incompleteness} of future null infinity $\mathcal I^+,$ which represents the set of idealized far away observers. It is well-known that such a singularity, denoted by $b_\Gamma$, can only form at the center $\{r=0\}$ and the past lightcone $C_0^-$ of $b_\Gamma$ naturally divides a naked singularity spacetime into an \emph{exterior-naked singularity region} $\mathcal R_{\mathrm{Ext}}$ and an \emph{interior fill-in}  $\mathcal R_{\mathrm{Int}}$, see \cref{fig:naked-singularity-intro}. It is precisely the geometry of the exterior region that determines the nakedness of the singularity, while the interior fill-in ensures that the spacetime arises dynamically from regular collapse.

The incompleteness of $\mathcal I^+$ of naked singularity spacetimes constitutes a failure of global existence for far away observers and threatens the predictability of general relativity. The weak cosmic censorship conjecture, originally formulated by Penrose in \cite{P69} for more general Einstein-matter systems, postulates that, \emph{generically}, naked singularities do not form and $\mathcal I^+$ is complete \cite{GH78,C99-GL}. In the seminal work \cite{C99}, Christodoulou proved the conjecture for \eqref{Einstein Scalar Field} in spherical symmetry. We shall return to this work in \cref{sec:intro-instability}.

\begin{figure}[h]
\centering{
\def\svgwidth{11pc}
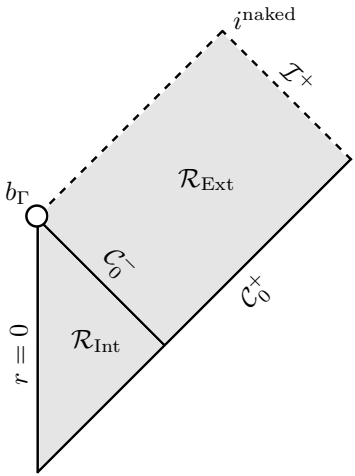}
\caption{Penrose diagram of a spacetime with future incomplete $\mathcal I^+$ and a naked singularity at $b_\Gamma$ arising from data on $C_0^+$. The past ingoing light cone of the singularity $b_\Gamma$ is denoted with $C_0^-$ and divides the spacetime into the exterior region $\mathcal R_{\mathrm{Ext}}$ and the interior fill-in $\mathcal R_{\mathrm{Int}}$.} 
\label{fig:naked-singularity-intro}
\end{figure}

The genericity condition in formulating the weak cosmic censorship conjecture is indeed necessary because naked singularity spacetimes have been shown to form dynamically for \eqref{Einstein Scalar Field} in spherical symmetry. We describe two kinds of naked singularities below, which can be classified as featuring \emph{continuous self-similarity} and \emph{discrete self-similarity}.
The continuously self-similar solutions were rigorously constructed, while the discretely self-similar solutions have been obtained only numerically.
We refer to \cref{fig:css-dss} for an illustration highlighting the differences between continuous and discrete self-similarity.

\begin{figure}[ht]
\begin{subfigure}{.48\textwidth}
\centering{
\def\svgwidth{12pc}
\begingroup%
  \makeatletter%
  \providecommand\color[2][]{%
    \errmessage{(Inkscape) Color is used for the text in Inkscape, but the package 'color.sty' is not loaded}%
    \renewcommand\color[2][]{}%
  }%
  \providecommand\transparent[1]{%
    \errmessage{(Inkscape) Transparency is used (non-zero) for the text in Inkscape, but the package 'transparent.sty' is not loaded}%
    \renewcommand\transparent[1]{}%
  }%
  \providecommand\rotatebox[2]{#2}%
  \newcommand*\fsize{\dimexpr\f@size pt\relax}%
  \newcommand*\lineheight[1]{\fontsize{\fsize}{#1\fsize}\selectfont}%
  \ifx\svgwidth\undefined%
    \setlength{\unitlength}{78.59154168bp}%
    \ifx\svgscale\undefined%
      \relax%
    \else%
      \setlength{\unitlength}{\unitlength * \real{\svgscale}}%
    \fi%
  \else%
    \setlength{\unitlength}{\svgwidth}%
  \fi%
  \global\let\svgwidth\undefined%
  \global\let\svgscale\undefined%
  \makeatother%
  \begin{picture}(1,1.18713699)%
    \lineheight{1}%
    \setlength\tabcolsep{0pt}%
    \put(0.08812372,0.32206218){\color[rgb]{0,0,0}\rotatebox{90}{\makebox(0,0)[lt]{\lineheight{1.25}\smash{\begin{tabular}[t]{l}$r=0$\end{tabular}}}}}%
    \put(0,0){\includegraphics[width=\unitlength,page=1]{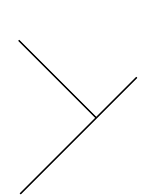}}%
    \put(0.0063832,0.9943573){\color[rgb]{0,0,0}\makebox(0,0)[lt]{\lineheight{1.25}\smash{\begin{tabular}[t]{l}$b_\Gamma$\end{tabular}}}}%
    \put(0.50687026,0.2704277){\color[rgb]{0,0,0}\rotatebox{45}{\makebox(0,0)[lt]{\lineheight{1.25}\smash{\begin{tabular}[t]{l}$\mathcal C_0^+$\end{tabular}}}}}%
    \put(0,0){\includegraphics[width=\unitlength,page=2]{naked-sing-CSS.pdf}}%
    \put(0.34933324,0.56317121){\color[rgb]{0,0,0}\rotatebox{-45}{\makebox(0,0)[lt]{\lineheight{1.25}\smash{\begin{tabular}[t]{l}$\mathcal C_0^-$\end{tabular}}}}}%
    \put(0,0){\includegraphics[width=\unitlength,page=3]{naked-sing-CSS.pdf}}%
  \end{picture}%
\endgroup%
}
\caption{Continuous self-similarity}
\end{subfigure}
\begin{subfigure}{.48\textwidth}
\centering{
\def\svgwidth{12pc}
\begingroup%
  \makeatletter%
  \providecommand\color[2][]{%
    \errmessage{(Inkscape) Color is used for the text in Inkscape, but the package 'color.sty' is not loaded}%
    \renewcommand\color[2][]{}%
  }%
  \providecommand\transparent[1]{%
    \errmessage{(Inkscape) Transparency is used (non-zero) for the text in Inkscape, but the package 'transparent.sty' is not loaded}%
    \renewcommand\transparent[1]{}%
  }%
  \providecommand\rotatebox[2]{#2}%
  \newcommand*\fsize{\dimexpr\f@size pt\relax}%
  \newcommand*\lineheight[1]{\fontsize{\fsize}{#1\fsize}\selectfont}%
  \ifx\svgwidth\undefined%
    \setlength{\unitlength}{78.59154168bp}%
    \ifx\svgscale\undefined%
      \relax%
    \else%
      \setlength{\unitlength}{\unitlength * \real{\svgscale}}%
    \fi%
  \else%
    \setlength{\unitlength}{\svgwidth}%
  \fi%
  \global\let\svgwidth\undefined%
  \global\let\svgscale\undefined%
  \makeatother%
  \begin{picture}(1,1.18713699)%
    \lineheight{1}%
    \setlength\tabcolsep{0pt}%
    \put(0.08812372,0.32206218){\color[rgb]{0,0,0}\rotatebox{90}{\makebox(0,0)[lt]{\lineheight{1.25}\smash{\begin{tabular}[t]{l}$r=0$\end{tabular}}}}}%
    \put(0.0063832,0.9943573){\color[rgb]{0,0,0}\makebox(0,0)[lt]{\lineheight{1.25}\smash{\begin{tabular}[t]{l}$b_\Gamma$\end{tabular}}}}%
    \put(0.50687026,0.2704277){\color[rgb]{0,0,0}\rotatebox{45}{\makebox(0,0)[lt]{\lineheight{1.25}\smash{\begin{tabular}[t]{l}$\mathcal C_0^+$\end{tabular}}}}}%
    \put(0,0){\includegraphics[width=\unitlength,page=1]{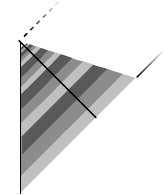}}%
    \put(0.29625709,0.62780191){\color[rgb]{0,0,0}\rotatebox{-45}{\makebox(0,0)[lt]{\lineheight{1.25}\smash{\begin{tabular}[t]{l}$\mathcal C_0^-$\end{tabular}}}}}%
    \put(0.3770897,0.95017086){\rotatebox{-17.35051561}{\makebox(0,0)[lt]{\lineheight{1.25}\smash{\begin{tabular}[t]{l}$\delta \log|u| = \Delta$\end{tabular}}}}}%
    \put(0,0){\includegraphics[width=\unitlength,page=2]{naked-sing-DSS.pdf}}%
  \end{picture}%
\endgroup%
}
\caption{Discrete self-similarity}
\end{subfigure}
\caption{Penrose diagram and illustration of the naked singularities with continuous and discrete self-similarity. The self-similarity is only depicted in a bounded region as the data on $C_0^+$ have to be truncated in order to obtain an asymptotically flat $C_0^+$.}
\label{fig:css-dss}
\end{figure}

\paragraph{Continuous self-similarity.}  
In \cite{C94}, Christodoulou constructed naked singularities based on a continuous self-similar ansatz for the solution. The singular nature of the spacetimes is a result of the blow-up of the scalar field towards $b_\Gamma$ along the ingoing self-similar cone $C_0^-$. The additional continuous symmetry reduces the PDE system \eqref{Einstein Scalar Field} to an autonomous ODE system, which plays an essential role in the construction. The proof is divided into two parts:
\begin{enumerate}
    \item[(i)] the construction of the exterior-naked singularity region $\mathcal R_{\mathrm{Ext}}$
    \item[(ii)] attaching an interior fill-in $\mathcal R_{\mathrm{Int}}$.
\end{enumerate}
In each region, the solution is obtained from a suitable orbit of the autonomous ODE system. A key feature of Christodoulou's examples is the limited regularity on the past cone of the singularity. In particular, the exterior-naked singularity region $\mathcal R_{\mathrm{Ext}}$ is already mildly irregular at $C_0^-$. Thus, when viewed in the context of the Cauchy problem, the data have an initial mild singularity at $C_0^- \cap C_0^+$. We refer to \cref{sec:intro-k-self-similar} for a more detailed discussion. These naked singularities have been generalized recently in \cite{S22} to solutions with an asymptotically continuous self-similar profile. The strategy of dividing the problem into (i) and (ii) is also used in the recent groundbreaking work \cite{RSR23, SR22} of Rodnianski and Shlapentokh-Rothman who constructed naked singularities for the Einstein vacuum equations, necessarily outside of spherical symmetry, for which the continuous self-similarity is generalized to a twisted self-similarity. We will discuss this work in more detail in \cref{sec:intro-k-self-similar}. 

\paragraph{Discrete self-similarity.} 
In the numerical study \cite{Cho93} on critical gravitational collapse for \eqref{Einstein Scalar Field}, Choptuik first observed naked singularities with asymptotically discretely self-similar symmetry to occur on the black hole formation threshold---the interface in the moduli space of solutions between dispersion and collapse. In these spacetimes, the singularity is caused by bounded infinite oscillations of the scalar field along the past light cone $C_0^-$ of the singular point $b_\Gamma$. Further numerical studies showed that all naked singularity solutions obtained using Choptuik's procedure asymptote to a universal solution with exact discretely self-similar symmetry as the singularity $b_{\Gamma}$ is approached, see the review article \cite{Gundlach2007}. We refer to this conjectured universal solution as the \textit{conjectured Choptuik solution}. Note that in view of the exact discrete self-similarity, the conjectured Choptuik solution would not itself be asymptotically flat, although it would occur as the limit of asymptotically flat solutions as the singularity is approached. 

\subsection{The main theorem: construction of exterior-naked singularity regions}
 
To date, no rigorous construction of discretely self-similar naked singularities has been achieved. As explained above, all examples to date have continuous self-similar symmetry and are of limited regularity. The difficulty arises from the fact that, unlike continuous self-similarity, imposing discrete self-similarity does not reduce the system to ODEs and one is still faced with a system of PDEs.

The main result of this paper is to address part (i) of the strategy outlined above and construct a class of exterior-naked singularity regions $\mathcal R_{\mathrm{Ext}}$ with discretely self-similar profile along $C_0^-$ arising from initial data which are moreover smooth away from the singularity.

\begin{theorem}\label{main theorem}
  For any $\Delta>0$, there exist spherically symmetric exterior-naked singularity solutions to \eqref{Einstein Scalar Field} in the region $\mathcal R_{\mathrm{Ext}} \doteq \{-1\leq u<0,\ v\geq0\}$ with a future incomplete $\mathcal{I}^+$ and a Penrose diagram as depicted in \cref{fig:main-thm}. In addition, the following hold.
  \begin{enumerate}
      \item The solutions are smooth in $\mathcal R_{\mathrm{Ext}}$, in particular up to and including the ingoing cone $C^-_0:=[-1,0)\times \{0\}$, i.e. they are smooth away from the singularity $b_\Gamma:=(0,0)$.
      \item Along the ingoing cone $C_0^-$, the solutions satisfy the discrete self-similarity properties 
\begin{equation}\label{self similar properties thm}
    \phi\circ  \theta_\Delta=\phi, \quad \mu\circ \theta_\Delta=\mu, \quad \frac{\partial_{v}  \phi}{\partial_v r} \circ \theta_\Delta = e^{\Delta} \frac{\partial_{v}\phi}{\partial_v r},
\end{equation}
where $\theta_\Delta (u,0) = (e^{-\Delta}u, 0)$ and the gauge condition $ r|_{C_0^-} = -u$ holds. 
\item The scalar field~$\phi$ and the mass aspect ratio $\mu= \frac{2m}{r} \doteq  1+4 \Omega^{-2} \partial_u r \partial_v r$ remain small in amplitude  \begin{equation}  \sup_{C_0^-} |\phi|\ll_{\Delta} 1, \quad  \sup_{C_0^-} |\mu| \ll_{\Delta} 1
\label{eq:smallness-hawking-mass}
\end{equation} but undergo infinite oscillations along $C_0^-$ as the singularity $b_\Gamma$ is approached.
  \end{enumerate} 
\begin{figure}[H]
\centering{
\def\svgwidth{11pc}
\begingroup%
  \makeatletter%
  \providecommand\color[2][]{%
    \errmessage{(Inkscape) Color is used for the text in Inkscape, but the package 'color.sty' is not loaded}%
    \renewcommand\color[2][]{}%
  }%
  \providecommand\transparent[1]{%
    \errmessage{(Inkscape) Transparency is used (non-zero) for the text in Inkscape, but the package 'transparent.sty' is not loaded}%
    \renewcommand\transparent[1]{}%
  }%
  \providecommand\rotatebox[2]{#2}%
  \newcommand*\fsize{\dimexpr\f@size pt\relax}%
  \newcommand*\lineheight[1]{\fontsize{\fsize}{#1\fsize}\selectfont}%
  \ifx\svgwidth\undefined%
    \setlength{\unitlength}{99.52730915bp}%
    \ifx\svgscale\undefined%
      \relax%
    \else%
      \setlength{\unitlength}{\unitlength * \real{\svgscale}}%
    \fi%
  \else%
    \setlength{\unitlength}{\svgwidth}%
  \fi%
  \global\let\svgwidth\undefined%
  \global\let\svgscale\undefined%
  \makeatother%
  \begin{picture}(1,0.99817224)%
    \lineheight{1}%
    \setlength\tabcolsep{0pt}%
    \put(0.79002683,0.84372848){\color[rgb]{0,0,0}\rotatebox{-45}{\makebox(0,0)[lt]{\lineheight{1.25}\smash{\begin{tabular}[t]{l}$\mathcal I^+$\end{tabular}}}}}%
    \put(0,0){\includegraphics[width=\unitlength,page=1]{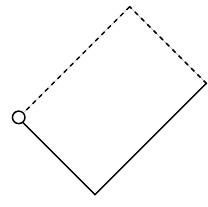}}%
    \put(-0.00091879,0.47907171){\color[rgb]{0,0,0}\makebox(0,0)[lt]{\lineheight{1.25}\smash{\begin{tabular}[t]{l}$b_\Gamma$\end{tabular}}}}%
    \put(0.70469565,0.16325318){\color[rgb]{0,0,0}\rotatebox{45}{\makebox(0,0)[lt]{\lineheight{1.25}\smash{\begin{tabular}[t]{l}$\mathcal C_0^+$\end{tabular}}}}}%
    \put(0.65613498,0.9705381){\color[rgb]{0,0,0}\makebox(0,0)[lt]{\lineheight{1.25}\smash{\begin{tabular}[t]{l}$i^{\mathrm{naked}}$\end{tabular}}}}%
    \put(0.42464003,0.00685392){\color[rgb]{0,0,0}\makebox(0,0)[lt]{\lineheight{1.25}\smash{\begin{tabular}[t]{l}$p$\end{tabular}}}}%
    \put(0,0){\includegraphics[width=\unitlength,page=2]{naked-sing-exterior.pdf}}%
    \put(0.27909655,0.31004515){\color[rgb]{0,0,0}\rotatebox{-45}{\makebox(0,0)[lt]{\lineheight{1.25}\smash{\begin{tabular}[t]{l}$\mathcal C_0^-$\end{tabular}}}}}%
  \end{picture}%
\endgroup%
}
\caption{Penrose diagram of an exterior-naked singularity region constructed in \cref{main theorem}. Along the ingoing cone $C_0^-$, the solution is discretely self-similar and satisfies \eqref{self similar properties thm}. The singularity at $b_\Gamma$ is characterized by the infinite but uniformly bounded oscillations of the scalar field $\phi$ and the mass aspect ratio $\mu$ towards $b_\Gamma$. }
\label{fig:main-thm}
\end{figure}
\end{theorem}

\begin{remark}
    The solutions constructed in \cref{main theorem} have an \emph{exactly} discretely self-similar profile only along the ingoing cone $C_0^-$ and are approximately discretely self-similar away from $C_0^-$. We note that constructing spacetimes with asymptotically flat null infinity $\mathcal I^+$ requires the data on $C_0^+$ to be asymptotically flat, which is incompatible with global discrete self-similarity. However, one should be able to produce non-asymptotically flat spacetimes with exact discrete self-similarity using the extracting procedure of \cite[Theorem~1.2]{RSR18}. These would still have the Penrose diagram of \cref{fig:main-thm}, but would not satisfy the usual notions of asymptotical flatness at $\mathcal{I}^+$, as for instance the Bondi mass would diverge there.
\end{remark}

\subsection{Attaching an interior fill-in and constructing a naked singularity from smooth data} 
Our main result provides a construction of an exterior-naked singularity region $\mathcal R_{\mathrm{Ext}},$ while it remains an open problem to construct an appropriate interior fill-in $\mathcal R_{\mathrm{Int}}$, in the past of $C_0^-$. 

\begin{openpb*}
Does any of the exterior-naked singularity regions constructed in \cref{main theorem} admit an interior fill-in $\mathcal{R}_{\mathrm{Int}}$ which can be attached smoothly to obtain a naked singularity spacetime $\mathcal{R}_{\rm Int}\cup\mathcal{R}_{\rm Ext}$ as depicted in \cref{fig:naked-singularity-intro}?
\end{openpb*}

We stress that our exterior-naked singularity regions are fundamentally different from those numerically observed by Choptuik in \cite{Cho93}. Specifically, the naked singularities in Choptuik's simulations are incompatible with the smallness condition \eqref{eq:smallness-hawking-mass}.

A natural way to construct an interior solution $\mathcal{R}_{\mathrm{Int}}$ is to make an exactly discrete self-similar ansatz in the interior region $\mathcal R_{\mathrm{Int}}$. Since the self-similar vector field (see already \eqref{eq:S-vector-field}) \begin{equation*}S=u\partial_u + (1-\kappa) v \partial_v \end{equation*}is timelike in $\mathcal R_{\mathrm{Int}}$, this requires constructing periodic-in-time solutions to hyperbolic PDEs. On the other hand, in the region $\mathcal R_{\mathrm{Ext}}$ which is successfully constructed in \cref{main theorem}, the vector field $S$ is spacelike, and we have to construct approximately periodic-in-space solutions. (For the construction in the exterior region there are robust techniques available, even outside spherical symmetry \cite{RSR23}, to which we return in \cref{sec:intro-k-self-similar}.)

An additional challenge in smoothly attaching an interior fill-in $\mathcal{R}_{\mathrm{Int}}$ to the exterior-naked singularity regions of \cref{main theorem} follows from the extension principle of Christodoulou \cite[Theorem~5.1]{C93}, applied to solutions that are at least $C^2$ regular. By the extension principle, for any such interior fill-in, the mass aspect ratio $\mu$ cannot satisfy the smallness bounds \eqref{eq:smallness-hawking-mass} globally on $\mathcal R_{\mathrm{Int}}$, and there would have to be a region in $\mathcal R_{\mathrm{Int}}$ where $\mu\sim 1$.

\subsection{The black hole formation threshold and the conjectured Choptuik solution}
We recall that the numerical simulations in \cite{Cho93} suggest that the black hole formation threshold consists of asymptotically discretely self-similar naked singularities which asymptote to a universal solution---the conjectured Choptuik solution. Based on these numerics we formulate the following open problem.

\begin{openpb*}
Does the black hole formation threshold consist of asymptotically discretely self-similar naked singularities? If so, do all such naked singularities asymptote to a universal discretely self-similar solution---the conjectured Choptuik solution---as the singularity is approached?
\end{openpb*}

A natural first problem in addressing the above is constructing a possible candidate for the conjectured Choptuik solution. Indeed, based on the numerical simulations in \cite{Cho93}, the work \cite{RT19} gives a computer-assisted construction of a smooth, discretely self-similar, spherically symmetric solution to \eqref{Einstein Scalar Field}, defined however only in an open neighborhood of the causal past of an ingoing cone $C_0^-$ with singular vertex.  If it can be proved that one can smoothly attach an exterior discretely self-similar solution, the construction would correspond to a 
natural candidate for the conjectured Choptuik solution.\footnote{By an easy truncation argument in the far region one could then produce asymptotically flat examples, and this would then show the formation of naked singularities from smooth asymptotically flat initial data.}

\subsection{Comparison with the \texorpdfstring{$k$}{k}-self-similar solutions and the vacuum case}
\label{sec:intro-k-self-similar}
We compare the solutions of \cref{main theorem} with the naked singularity spacetimes of Christodoulou \cite{C94} and the examples of Rodnianski and Shlapentokh-Rothman \cite{RSR23,SR22}. Since we construct only exterior-naked singularity regions in \cref{main theorem}, we restrict our comparison to the exterior regions in \cite{C94} and \cite{RSR23} and the behavior on the past light cone of the singularity.

As discussed before, the $k$-self-similar spacetimes constructed in \cite{C94} are continuously self-similar solutions of \eqref{Einstein Scalar Field} and have naked singularities for a certain range of the parameter $k.$ Along the ingoing cone $C_0^-$, the   $k$-self-similar solutions satisfy
\begin{equation*}
    \phi=kt+O(1),\ \mu=\frac{k^2}{1+k^2},
\end{equation*}
where $t\doteq -\log|u|$ denotes the self-similar time. Thus, the singularity is a consequence of the blow-up of the scalar field and the failure of the mass aspect ratio to converge to zero. This is in contrast to the solutions constructed in \cref{main theorem} for which the scalar field and mass aspect ratio are periodic (see \eqref{self similar properties thm}) along $C_0^-$  in terms of the self-similar variable $t$ and the singularity is a consequence of this infinite oscillation.

In the works \cite{RSR23,SR22}, the first examples of the dynamical formation of naked singularities for the Einstein vacuum equations were constructed. The solutions have a twisted self-similar symmetry along $C_0^-$, which represents a novel continuously self-similar mechanism for singularity formation, see also \cite{Shl23}. The construction of the exterior-naked singularity regions in \cite{RSR23} is similar to that in \cref{main theorem} in many aspects. In both spacetimes the singularity is caused by a self-similar profile on $C_0^-,$ with the infinite twisting of the geometry in \cite{RSR23} being replaced by the infinite oscillation in \cref{main theorem}. However, the exterior regions are not assumed to have an exact self-similar symmetry in the future of $C_0^-$, and they turn out to be only asymptotically self-similar. For both constructions the strategy is to first prescribe characteristic initial data compatible with the self-similar symmetry on $C_0^-$ and then prove global existence of the solution. We note that the guideline of \cite{RSR23} is the natural approach to construct exterior-naked singularity regions, in the case where the system of equations does not reduce to ODEs as in \cite{C94}. 

In terms of the regularity of the spacetimes, there is a stark contrast between the solutions of \cref{main theorem} and the examples of \cite{C94} and \cite{RSR23,SR22}. In agreement with the numerical heuristics for discretely self-similar spacetimes, the solutions of \cref{main theorem} are globally smooth, except for $\{u=0\}$ of course. In particular, the solutions are smooth at $C_0^-.$ On the other hand, for the $k$-self-similar and twisted self-similar examples of naked singularities, the construction of the exterior region already restricts the regularity of the spacetime to $C_v^{1,\gamma}$ for some $\gamma >0$ at $C_0^-$, while the solution is smooth away from $C_0^-$.

The solutions of \cref{main theorem} arise from small data by construction. Even though the infinite oscillation of the mass aspect ratio $\mu$ on $C_0^-$ causes the singularity, we obtain that $\mu$ is everywhere small, which is essential in the proof of global existence. On the other hand, the solutions of \cite{C94} do not have any restrictions on the smallness of data, and global existence is obtained by the study of an autonomous ODE system. Finally, the solutions of \cite{RSR23} are in a large data regime on $C_0^-$. Global existence is achieved since the initial data on $C_0^+$ decays rapidly away from $C_0^-$, but this is done at the cost of the aforementioned limited Hölder regularity along $C_0^-$.

\subsection{The instability of naked singularities}
\label{sec:intro-instability}
In \cite{C99}, Christodoulou showed that \emph{all} naked singularity solutions to \eqref{Einstein Scalar Field} are unstable, thereby proving the weak cosmic censorship conjecture for the Einstein-scalar field system in spherical symmetry. The precise formulation of his result required a low regularity class of initial perturbations. The argument crucially relies on rough perturbations that trigger a blue-shift instability along the cone $C_0^-,$ which inevitably leads to the formation of trapped surfaces. Moreover, we note that the continuously self-similar naked singularity solutions in \cite{C94} belong to the low regularity class considered, so the genericity assumption in the formulation of weak cosmic censorship in the low regularity class is indeed necessary. In the smooth class, both the question of the formation of naked singularities and the question of their genericity remain open.
In view of the recent linear analysis in \cite{S24}, the latter appears to be a very subtle and intricate problem.

\subsection{Outline of the proof of \texorpdfstring{\cref{main theorem}}{Theorem 1}}

In the following, we will outline the main ideas used in the proof of \cref{main theorem}. 
We begin by writing down the system of equations \eqref{Einstein Scalar Field} in spherical symmetry expressed in double null coordinates which reduces to the wave equations 
\begin{align}
    \label{partial uv Omega}
\partial_u\partial_{v}\log\Omega^2 & = -2\partial_u\phi\partial_{v}\phi+\frac{\Omega^2}{2r^2}+\frac{2\partial_{v}r\partial_ur}{r^2}\\ 
\label{partial uv r}
    \partial_u(r\partial_{v}r)&=-\frac{1}{4}\Omega^2
\\ \label{wave}
r\partial_u\partial_{v}\phi  &= - \partial_{v}\phi\partial_ur-\partial_u\phi\partial_{v}r 
\end{align}
and the Raychaudhuri equations
\begin{align}\label{partial u}
    \partial_u(\Omega^{-2}\partial_ur)&=-r\Omega^{-2}(\partial_u\phi)^2\\
\label{partial v}
    \partial_{v}(\Omega^{-2}\partial_vr)&=-r\Omega^{-2}(\partial_{v}\phi)^2.
\end{align}
Moreover, the \emph{Hawking mass} \begin{equation*}
    m \doteq \frac{r}{2} (1+4 \Omega^{-2} \partial_v r \partial_u r)
\end{equation*}
satisfies 
\begin{equation}\label{eq:hawking-mass-v}
\partial_v m = 2r^2\Omega^{-2}(- \partial_ur) (\partial_v\phi)^2.
\end{equation}
\paragraph{Discretely self-similar characteristic data (\cref{initial data section}).}
We will set up smooth characteristic initial data on the cones $C_0^-\cup_p C_0^+ = \{ -1\leq u <0 \}\times \{ v=0\} \cup  \{ u=-1\}\times \{ v\geq 0\}$, which intersect at the sphere located at $p= (-1,0)$.
Characteristic data for the system \eqref{Einstein Scalar Field} consist of prescribing $r,\Omega,\phi$ which satisfy \eqref{partial u} on $C_0^-$ and \eqref{partial v} on $C_0^+,$ together with suitable gauge conditions, see \cref{sec:gauge-conditions}.

The starting point in our construction is to consider smooth, discretely self-similar functions $\phi: C_0^- \times \{ 0\} \rightarrow\mathbb{R}$ satisfying
\begin{equation}\label{ansatz for integral phi}
    \kappa\Delta=\int_{-1}^{-e^{-\Delta}}(-u')\big(\partial_u\phi\big)^2du'
\end{equation}for some small constant $0<\kappa\ll1$ as well as the additional pointwise bounds
\begin{equation}\label{ansatz for phi C0-}
    \phi(e^{-\Delta}u,0)=\phi(u,0),\ |\phi(u,0)|\lesssim\sqrt{\kappa},\ |\partial_u\phi(u,0)|\lesssim\sqrt{\kappa}/|u|.
\end{equation}
The parameter $\kappa$ in \eqref{ansatz for integral phi} can be seen as the analog of $k^2$ in the $k$-self-similar spacetimes in \cite{C94}. 

The lapse function $\Omega^2$ is determined by $\phi$ on $C_0^-$ according to \eqref{partial u}, and  satisfies the scaling relation
\begin{equation}
    \Omega^2(e^{- \Delta n}u,0) =e^{-n\Delta\kappa} \Omega^2(u,0).
\end{equation}
We remark that the scaling of $\Omega^2$ implies that the underlying vector field which determines the self-similar scaling \eqref{self similar properties thm} is given by
\begin{equation}\label{eq:S-vector-field}
S= u\partial_u + (1-\kappa) v \partial_v.
\end{equation}

Enforcing the discrete self-similarity \eqref{self similar properties thm} along $C_0^-$ introduces \emph{rigidity} conditions on the outgoing data on $C_0^+$ to first order at $p$. In particular, from the discrete self-similarity property \eqref{self similar properties thm}, equations \eqref{partial uv r} and \eqref{wave}, we obtain that the transverse derivatives $\partial_v r$ and $\partial_v\phi$ at $p$ are completely determined by the data on $C_0^-$.\footnote{We are only concerned with discrete self-similarity to first order at  $C_0^-$, but higher order versions of \eqref{self similar properties thm} would lead to higher order rigidity for the transverse derivatives at $p$.}
An essential aspect of our discretely self-similar construction is that the rigidity condition give us the smallness properties
\begin{equation}
\big|\partial_v\phi(p)\big|\lesssim\sqrt{\kappa}, \quad     \mu(p)\lesssim\kappa.
\end{equation}
This is in contrast with the $k$-self-similar spacetimes in \cite{C94}, where the outgoing data $\partial_v\phi(p)$ is of size $1/k,$ which is never small in the range of the parameter $k$. At higher order on $C_0^+$, we have the freedom to prescribe $\phi$ smoothly\footnote{Our proof also works if $\phi$ is at least $C^{1,\gamma}_{v}\big([0,\infty)\big)$ regular, where $\gamma=2\kappa/(1-\kappa).$ This is the same regularity exponent as in \cite{RSR23}.}, provided that certain smallness properties and asymptotically flat conditions hold.

\paragraph{Construction of the exterior region $\mathcal R_{\mathrm{Ext}}$ (\cref{sec:exterior-region}).} Once we set up the characteristic initial data, we prove that it leads to a global solution in the exterior region $\mathcal R_{\mathrm{Ext}}$. Our argument for global existence follows the strategy of \cite{RSR23} and relies only on the smallness of the initial data and bounds compatible with self-similarity, without requiring the solution to be exactly self-similar. In our proof, we decompose the spacetime into certain regions, and use bootstrap arguments to show suitable pointwise estimates for the solution in each region. In the last step, these bounds are then used to prove that $\mathcal{I}^+$ is incomplete, so the spacetime constructed represents indeed an exterior-naked singularity region.

We consider a small constant $\underline{v}$ such that $0<\kappa\ll\underline{v}\ll 1$, and the regions depicted in \cref{fig:regions}
{\allowdisplaybreaks\begin{align*}\mathcal R_I&=\bigg\{(u,v) \in \mathcal R_{\mathrm{Ext}}: v\in[0,\underline{v}^{1-\kappa}],\ \frac{v}{|u|^{1-\kappa}}\leq\underline{v}^{1-\kappa}\bigg\}\\ 
\mathcal R_{II}&=\bigg\{(u,v)\in \mathcal R_{\mathrm{Ext}}:  v\in[0,\underline{v}^{-1+\kappa}],\ \underline{v}^{1-\kappa}\leq\frac{v}{|u|^{1-\kappa}}\leq\underline{v}^{-1+\kappa}\bigg\}\\
\mathcal R_{III}&=\bigg\{(u,v) \in \mathcal R_{\mathrm{Ext}}:  v\in[0,1],\ \frac{v}{|u|^{1-\kappa}}\geq\underline{v}^{-1+\kappa}\bigg\}\\
\mathcal R_{IV}&=\bigg\{(u,v) \in \mathcal R_{\mathrm{Ext}}:  v\in[1,\infty),\ \frac{v}{|u|^{1-\kappa}}\geq\underline{v}^{-1+\kappa}\bigg\}.
\end{align*}}

\begin{figure}[ht]
\centering{
\def\svgwidth{15pc}
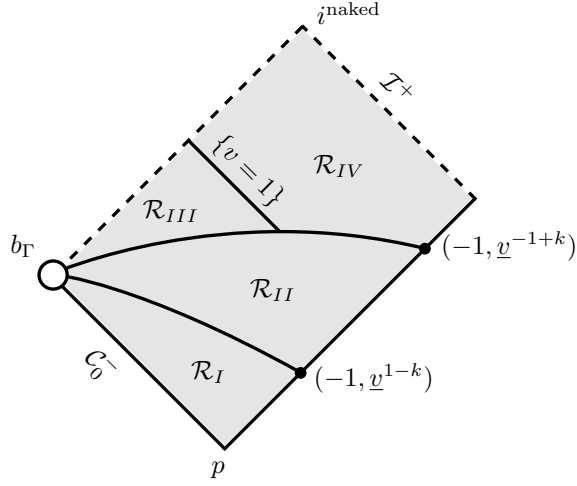}
\caption{The decomposition of $\mathcal R_{\mathrm{Ext}}$ into the regions $\mathcal R_I$--$\mathcal R_{IV}$}
\label{fig:regions}
\end{figure}

In region $\mathcal R_I$, we propagate the pointwise bounds on the solution on the ingoing cone $C_0^-,$ using the smallness of $\underline{v}.$ The estimates are straightforward for the quantities satisfying transport equations in the $\partial_v$ direction, such as $\Omega^{-2}\partial_vr,\ \partial_u\phi,\ \partial_ur,$ and $\partial_u\log\Omega^2.$ On the other hand, in order to bound $\Omega^{-2}\partial_v\phi$ we integrate the $\partial_u$ equation satisfied by $\Omega^{-2}\partial_v\phi(u,v)-\Omega^{-2}\partial_v\phi(u,0).$ 

In regions $\mathcal R_{II}$, $\mathcal R_{III}$, and $\mathcal R_{IV}$, we consider the regular change of coordinates $V=v^{\frac{1}{1-\kappa}}$ motivated by the  scaling vector field $S$ in \eqref{eq:S-vector-field}, and we denote the new lapse function by $\bar{\Omega}.$ We prove estimates in region $\mathcal R_{II}$ by conjugating the equations with $\exp\big(D\cdot V/u\big)$, for some large constant $D.$ These estimates are analogous to applying a local-in-time energy estimate in the ``time'' variable $V/u$, with data on the spacelike hypersurface $\big\{V=-\underline{v}u\big\}.$  

In region $\mathcal R_{III}$ we prove pointwise estimates using the smallness of $\underline{v}$ for the quantities satisfying transport equations in the $\partial_u$ direction, such as $\bar{\Omega}^{-2}\partial_ur,\ \partial_V\phi,\ \partial_Vr,$ and $\partial_V\log\bar{\Omega}^2.$ However, in order to bound $\partial_u\phi,$ we conjugate the equation for $r\partial_u\phi$ by $|u|^p/V^p,$ for some small constant $p>0.$ In particular, we obtain that $\phi \in C_u^{1-p}$, but we cannot conclude that $\partial_u\phi$ and $\bar{\Omega}$ extend continuously to $\{u=0\}.$ Nevertheless, the quantitative bounds obtained will suffice for our purposes. 

We prove estimates in region $\mathcal R_{IV}$ using the a~priori bound on the Hawking mass obtained from the initial data bound and the monotonicity properties. It is convenient to split region $\mathcal R_{IV}$ into two sub-regions: $\mathcal R_{IV}\cap\{v\leq\underline{v}^{-1+\kappa}\}$ with similar pointwise bounds to those of region $\mathcal R_{III}$, and $[-1,0)\times[\underline{v}^{-1+\kappa},\infty)$ with bounds specific to an asymptotically flat region. Finally, we then use the estimates in region $\mathcal R_{IV}$ to show that $\mathcal{I}^+$ is incomplete.
\medskip

\textbf{Acknowledgments.}
The authors would like to express their gratitude to   Igor Rodnianski for many stimulating and valuable discussions. The authors also thank Mihalis Dafermos for useful conversations.  

\section{Discretely self-similar characteristic data}\label{initial data section}

The characteristic initial data are given by $r,\Omega,\phi$ which satisfy Raychaudhuri's equations \eqref{partial u} on $C_0^-$ and \eqref{partial v} on $C_0^+.$ In this section we set up suitable characteristic initial data, which are consistent with ingoing discrete self-similarity along $C_0^-$ and asymptotic flatness along $C_0^+.$

\subsection{Underlying manifold and gauge conditions}
\label{sec:gauge-conditions} 

We will consider the maximal globally hyperbolic development of data imposed on the bifurcate double null cones $C_0^- = \{ -1 \leq u < 0\} \times  \{ v=0\} $ and $C_0^+ = \{  u =-1\} \times  \{ 0 \leq v < \infty \} $. The bifurcation sphere is $p = \{ u=-1\} \times \{ v=0\}.$ We will denote the future domain of dependence of $C_0^- \cup_p C_0^+$ with $\mathcal R_{\mathrm{ext}} = \{ -1 \leq u <0\} \times \{ 0\leq v <\infty \}$.

On the incoming cone $C_0^-$ we impose 
\begin{equation*}
    r(u,0) = -u.
\end{equation*}
This can be interpreted as choosing the gauge condition $\partial_u r = -1$ along $C_0^-$ and $r(p)=1$. We also impose that   $\phi(p)=0$.

On the outgoing cone $C_0^+$ we impose the gauge \begin{equation}\label{eq:gauge-for-omega-outgoing}
\partial_v \Omega^2 =0.
\end{equation}
This fixes the gauge upon specifying $\Omega^2(p)$ which will be done later. We are now in the position of defining characteristic data with ingoing discrete self-similarity.

    \subsection{Characteristic data with ingoing discrete self-similarity}
\begin{definition}[Characteristic data with ingoing discrete self-similarity]\label{def:discrete-self-similar}
    Let $N > 0$ be a regularity index for the tangential derivatives. We say that the seed $\phi(\cdot,0):C_0^- \to\mathbb{R}$, $\phi(0,\cdot)\colon C_0^+ \to \mathbb R$ with $\phi(p)=0$ gives rise to a $C^N$-regular characteristic data set on $C_0^- \cup_p C_0^+$ with \emph{ingoing discrete self-similarity} if there exist constants  $\Delta>0$, $0<\kappa<1$, and $C_{data}>0$ such that upon specifying \begin{equation}\label{eq:Omega-on-vertex}
\Omega^2(p)=\frac{1}{1-\kappa}\end{equation} and imposing the gauge conditions from \cref{sec:gauge-conditions}, the following conditions hold:
\begin{enumerate} 
\item The scalar field satisfies $\phi (\cdot, 0) \in C^{N+1} (C_0^-)$ and $\phi (0,\cdot) \in C^{N+1} (C_0^+)$.
\item The scalar field is discretely self-similar on $C_0^-$, i.e., \begin{equation}\label{eq:self-similar-property-phi}
\phi(e^{-\Delta}u,0) =\phi(u,0).\end{equation}
\item The constant $\kappa$ satisfies
\begin{equation}
    \label{eq:constraint-on-phi}
\kappa  = \frac{F(-e^{-\Delta})}{\Delta},
\end{equation}
where 
\begin{equation}\label{eq:definition-F}
   F(u) \doteq  \int_{-1}^{u}(-u')\big(\partial_u\phi(u',0)\big)^2du'.
\end{equation}
\item The scalar field satisfies the self-similar bounds \begin{equation}\label{eq:bounds-on-incoming-cone}
    \ |\phi(u,0)|\leq C_{data} \sqrt{\kappa},\ |\partial_u\phi(u,0)|\leq  C_{data} \sqrt{ \kappa }|u|^{-1} 
\end{equation}
on $C_0^-$ .
\item The mass aspect ratio $\mu \doteq  \frac{2m}{r} $ is discretely self-similar on $C_0^-$, i.e., 
\begin{equation} \label{eq:self-similarity-for-mu}
    \mu(e^{-\Delta}u,0) =\mu(u,0).
\end{equation}
\item  The first transverse derivatives of the scalar field and of the radiation field are discretely self-similar along $C_0^-$
\begin{equation} \label{eq:self-similarity-outgoing-der}
 \frac{\partial_v (r\phi ) }{\partial_v r  } (e^{-\Delta}u,0) = \frac{\partial_v (r\phi)}{\partial_v r  }  (u,0)
\end{equation}
and \begin{equation} \label{eq:self-similarity-outgoing-der-2}
\frac{\partial_v \phi }{\partial_v r  } (e^{-\Delta}u,0) =e^{\Delta} \frac{\partial_v \phi }{\partial_v r  }  (u,0).
\end{equation}
\item The scalar field satisfies the asymptotically flat conditions on $C_0^+$ \begin{equation}\label{ansatz for phi C0+}
    \big|\partial_v\phi(-1,v)\big|\leq C_{data}\sqrt{\kappa}\cdot\frac{1}{(1+v)^2}, \quad     \big|\partial_v^2 \phi(-1,v)\big|\leq C_{data}\sqrt{\kappa}\cdot\frac{1}{(1+v)^{3}}.
\end{equation}
\end{enumerate}
\end{definition}

\begin{remark}
    The equation \eqref{eq:constraint-on-phi} can also be regarded as a definition of  $\kappa$ which measures the size of $\phi$ along $C_0^-$. In our construction, we will take $\kappa$ to be a small quantity. 
    \end{remark}

\begin{remark}
    Note that $\Omega^2$ and the transverse derivatives $\partial_v r$ and  $\partial_v \phi$ along $C_0^-$ are determined by their values at the bifurcation sphere $p$ and \eqref{partial uv r}--\eqref{partial u}. We will prove in \cref{lem:derivation-of-c1} and \cref{lem:derivation-of-c2} in \cref{subsec:consequences-of-discrete-self-similarity} that for discretely self-similar data as in \cref{def:discrete-self-similar}, the values  $\partial_v r(p)$ and $\partial_v \phi(p)$ are rigid.   
\end{remark}

The statement that there exist non-trivial characteristic data with ingoing discrete self-similarity as defined in \cref{def:discrete-self-similar} is the content of the following proposition, which will be proved in \cref{sec:proof-nontrivial-data}.
\begin{proposition}\label{prop:data-non-trivial}
    For any $\Delta >0$,  there exists a $C_{data}=C_{data}(\Delta)>0$ such that for any $\kappa>0$, and $N\geq 1$,  there exists a non-trivial seed function $\phi \colon C_0^- \cup_p C_0^+ \to \mathbb R$ giving rise to $C^k$-regular characteristic data with ingoing self-similarity as defined in \cref{def:discrete-self-similar}.
\end{proposition}

\subsection{Consequences of imposing discrete self-similarity}
\label{subsec:consequences-of-discrete-self-similarity}

 In \cref{initial data section} we use the following convention. For non-negative quantities $A$ and $B$, we write $A\lesssim B$ if there exists a constant $\tilde C>0$ only depending on $C_{data}$ and $\Delta$ from  \cref{def:discrete-self-similar} such that $A\leq \tilde C B$.  We will also write  ``for sufficiently small $0< \kappa \ll 1 $'' if there exists a constant $\kappa_0>0$ sufficiently small and only depending on $C_{data}$, $\Delta$  from  \cref{def:discrete-self-similar} such that $\kappa$ satisfies $0<\kappa \leq \kappa_0$.

\subsubsection{Estimates and rigidity along the ingoing cone $C_0^-$} 

As a consequence of discrete self-similarity along $C_0^-$, we prove rigidity for the lapse $\Omega^2$ in \cref{lem:properties-of-F} and \cref{lem:rigidity-of-Omega}. We also prove rigidity for $\partial_vr$ and $\partial_v\phi$ on $C_0^-$ in \cref{lem:derivation-of-c1} and \cref{lem:derivation-of-c2}, which in particular determines the outgoing data on $C_0^+$ to first order at $p.$

\begin{lemma}\label{lem:properties-of-F}
Consider discretely self-similar data as in  \cref{def:discrete-self-similar}.
    Then, the function $F$ as defined in \eqref{eq:definition-F} satisfies 
\begin{equation}
    F(-1) =0, \quad F(e^{- n\Delta} u ) = F(u) +  n \kappa\Delta.\label{eq:relation-for-F}
\end{equation}
for $n \in \mathbb N$ and 
\begin{equation}\label{eq:bound-on-F}
  \sup_{u \in[-1,0)}  \Big|\frac{F(u)}{\log|u|}\Big|\lesssim \kappa, \quad\left|\frac{d}{du} F(u)\right|\lesssim \kappa |u|^{-1},\ \left|\log\big(|u|^{-\kappa}e^{-F(u)}\big)\right|\lesssim\kappa.
\end{equation}
\begin{proof}
    We compute \begin{align*}F(u) & = \int_{-1}^{ u } (-u') (\partial_u \phi  )^2 du' = \int_{- e^{-\Delta} }^{ e^{-\Delta} u } (- e^\Delta u') (e^{-\Delta} \partial_u \phi)^2 e^{\Delta} du' \\
    & = \int_{-1}^{e^{-\Delta}u } (-u') (\partial_u\phi)^2 du' - \int_{-1}^{-e^{-\Delta} } (-u') (\partial_u\phi)^2 du' = F(e^{-\Delta}u) -  \kappa\Delta,\end{align*}
    where we used a change of variables,  \eqref{eq:self-similar-property-phi} and \eqref{eq:constraint-on-phi}. Iterating this yields \eqref{eq:relation-for-F}. The bounds \eqref{eq:bound-on-F} are a direct consequence of \eqref{eq:bounds-on-incoming-cone} and \eqref{eq:relation-for-F}.
\end{proof}
\end{lemma}

\begin{lemma}\label{lem:rigidity-of-Omega} Consider self-similar data as in \cref{def:discrete-self-similar}. Then,
the lapse $\Omega^2(u,0)$ on $C_0^-$ is uniquely determined as
    \begin{equation}\label{eq:Omega-on-cin}
      \Omega^{2}(u,0) =  \frac{1}{1-\kappa} e^{-F(u)},
    \end{equation}
    where $F$ is defined in \eqref{eq:definition-F}.
   In particular, 
    \begin{equation}
        \Omega^2(e^{- \Delta n}u,0) =e^{-n\Delta\kappa} \Omega^2(u,0).
    \end{equation}
    \begin{proof}
By our gauge condition on $C_0^{-}$ we have $\partial_u r =-1$. Thus, \eqref{partial u} reads 
\begin{equation*}
    \partial_u (\Omega^{-2} ) = -u \Omega^{-2} (\partial_u \phi)^2
\end{equation*}
or
\begin{equation}\label{eq:log-Omega}
    \partial_u \log \Omega^{2} = u (\partial_u \phi)^2 
\end{equation}
which shows \eqref{eq:Omega-on-cin} upon integration and using \eqref{eq:definition-F}. 
    \end{proof}
\end{lemma}

\begin{lemma}
\label{lem:derivation-of-c1}
Consider discretely self-similar data as in  \cref{def:discrete-self-similar}. Then, the following holds:
\begin{enumerate}
    \item The quantity $\Omega^{-2}\partial_vr $ is discretely self-similar on $C_0^-$: 
    \begin{equation}\label{dv r DSS}
\Omega^{-2}\partial_vr(u,0)=\Omega^{-2}\partial_vr(e^{-\Delta}u,0).
    \end{equation}
    \item   The transverse derivative $\Omega^{-2} \partial_v r(p)$ at the bifurcation sphere $p$ is determined by the data on $C_0^-$ and satisfies   
 \begin{equation} \label{eq:rigidity-of-partial_vr}
 \Omega^{-2} \partial_v r(p) =\frac{1}{4-4e^{-\Delta(1+\kappa)}}\int_{-1}^{-e^{-\Delta}}   e^{-F(u')} du'.\end{equation}
 \item  The transverse derivative $\Omega^{-2} \partial_v r(u,0)$ along $C_0^-$ is determined by the data on $C_0^-$ and given as
    \begin{equation}\label{dv r formula}
        \Omega^{-2}\partial_vr(u,0)=\frac{e^{F(u)}}{-u}\cdot\bigg(\Omega^{-2} \partial_v r(p)-\frac{1}{4}\int_{-1}^{u}e^{-F(u')}du'\bigg).
    \end{equation}
    \item  For $0<\kappa\ll 1$ sufficiently small, we have the bounds
    \begin{align}\label{dv r bound}
        \sup_{u\in [-1,0)}\big|4\Omega^{-2}\partial_vr-1\big| \lesssim \kappa, \quad   \sup_{u\in [-1,0)}|\mu|\lesssim  \kappa.
    \end{align}
\end{enumerate}
   
\end{lemma}
\begin{proof}
    The constraint \eqref{dv r DSS} follows from \begin{equation}\label{eq:computation-mu} \mu = \frac{2m}{r} = 1 +4 \Omega^{-2} \partial_u r \partial_v r = 1 - 4 \Omega^{-2} \partial_v r \end{equation} on $C_0^-$ and the imposed discrete self-similarity condition \eqref{eq:self-similarity-for-mu} for $\mu$. 

    We now use \eqref{wave}, \eqref{eq:log-Omega} and $\partial_u r =-1$  to compute 
    \begin{align*}
        \partial_u (\Omega^{-2} \partial_v r)&= \partial_u \Omega^{-2} \partial_v r  + \Omega^{-2} \left( \frac{-\Omega^2}{4r} - \frac{\partial_u r \partial_v r }{r}\right)\\
        & = \Omega^{-2} \partial_v r \left( \partial_u \log \Omega^{-2} - \frac{\partial_ur}{r}\right) - \frac{1}{4r} \\ & = \Omega^{-2} \partial_v r \left( -u (\partial_u \phi)^2 -  \frac{1}{u}\right) + \frac{1}{4u}.
    \end{align*}
  Using the definition of $F$ we rewrite the above as 
    \begin{equation}
   \partial_u     \left( e^{ - F(u) +\int_{-1}^u \frac{1}{u'} d u' } \Omega^{-2} \partial_v r
 \right) = \frac{1}{4u} e^{ - F(u) +\int_{-1}^u  \frac{1}{u'} d u' } = -\frac 14 e^{ - F(u)  }.
    \end{equation}
 We integrate from $u=-1$ to $u$ to obtain     \begin{equation}\label{eq:computation-omega-partialvr}
        (-u)  e^{-F(u)} \Omega^{-2}\partial_v r (u) - \Omega^{-2} \partial_v r(p)= \int_{-1}^{u} -\frac 14 e^{-F(u')} d u'.
    \end{equation}
    Solving for $\Omega^{-2} \partial_v r $ yields \eqref{dv r formula}. In order to obtain the formula for $\Omega^{-2} \partial_v r(p)$ we use the periodicity of $\Omega^{-2} \partial_v r $ and obtain from \eqref{eq:computation-omega-partialvr} that $\Omega^{-2} \partial_v r(p)$ has to satisfy
    \begin{equation}
      \Omega^{-2} \partial_v r(p) ( e^{ \Delta (-\kappa -1)    }  -1) = \int_{-1}^{-e^{-\Delta}} - \frac 14 e^{-F(u')} du'.
    \end{equation}
    Solving for $\Omega^{-2} \partial_v r(p)$ gives \eqref{eq:rigidity-of-partial_vr}.

    To prove the bounds in \eqref{dv r bound}, we first observe the bounds \begin{equation*}\sup_{u\in [-1,-e^{-\Delta}]} |F|\lesssim \kappa, \quad  \sup_{u\in [-1,-e^{-\Delta}]} |e^{-F}-1|\lesssim \kappa,\quad \text{ and } \quad|e^{-\Delta \kappa} -1 |\lesssim \kappa.\end{equation*}
    Thus, using \eqref{eq:rigidity-of-partial_vr} we obtain 
    \begin{equation*}
        \left|4 \Omega^{-2} \partial_v r (p) - 1\right| = \left| \frac{1}{1-e^{-\Delta -\Delta \kappa}} \int_{-1}^{-e^{-\Delta}} e^{F(u')} du' - 1 \right| \lesssim \left|\frac{1}{1-e^{-\Delta}} \int_{-1}^{-e^{-\Delta}}  du'  - 1\right| +  \kappa = \kappa. 
    \end{equation*}
    Using the above estimates, \eqref{dv r formula} and \eqref{dv r DSS} we obtain that 
    \begin{align*}
 \sup_{u\in [-1,0)}|     4\Omega^{-2} \partial_v r(u,0) -1 |  & =  \sup_{u\in [-1,-e^{-\Delta}]} |     4\Omega^{-2} \partial_v r(u,0) -1 | \\
 & = \sup_{u\in [-1,-e^{-\Delta}]}  \frac{e^{F(u)}}{-u} \left| 4 \Omega^{-2} \partial_v r(p) - \int_{-1}^u e^{-F(u')} du' +u e^{-F(u)} \right| \lesssim \kappa.
    \end{align*}
The bound on $\mu$ follows from the above and   \eqref{eq:computation-mu} which concludes the proof.  
\end{proof}

Analogously to the rigidity for $\Omega^{-2} \partial_v r (p)$ from \cref{lem:derivation-of-c1} will now show that the discrete self-similarity for $\frac{\partial_v (r\phi)}{\partial_v r  }$ along $C_0^-$ constrains $\Omega^{-2}\partial_v \phi$ at the bifurcation sphere $p$. 

 \begin{lemma}
 \label{lem:derivation-of-c2}
Consider discretely self-similar data as in  \cref{def:discrete-self-similar}. Then, the following holds:
\begin{enumerate}
\item      The quantity $\Omega^{-2} \partial_v(r\phi) $ is discretely self-similar on $C_0^-$
    \begin{equation}\label{dv-phi-DSS}
\Omega^{-2} \partial_v(r\phi)(u,0)=\Omega^{-2} \partial_v(r\phi)(e^{-\Delta}u,0).
    \end{equation}
       \item   The transverse derivative $ \Omega^{-2} \partial_v(r\phi)(p) $ at the bifurcation sphere $p$ is determined by the data on $C_0^-$ and satisfies   
\begin{equation} \label{eq:value-of-c2}
   \Omega^{-2} \partial_v(r\phi)(p) = - \frac{ e^{\kappa\Delta}}{ 1- e^{  \kappa\Delta }} \int_{-1}^{-e^{-\Delta}}    \frac{ e^{-F(u')}}{-4 u'} \mu(u') \phi (u')du '.
\end{equation}
 \item  The transverse derivative $ \Omega^{-2}\partial_v (r\phi) (u,0)$  along $C_0^-$ is determined by the data on $C_0^-$ and given as
\begin{equation}\label{eq:definition-of-omega-dvrphi}
    \Omega^{-2} (u,0) \partial_v (r\phi) (u,0) =  e^{F(u)} \left(\Omega^{-2} \partial_v(r\phi)(p)  -   \int_{-1}^u \frac{e^{-F(u')}}{-4u'} \mu(u') \phi(u') du'\right).
\end{equation}
    \item  
For  $0<\kappa \ll 1$ sufficiently small, we have the bound
    \begin{equation}\label{dv phi bound}
        \big|\Omega^{-2}\partial_v(r\phi)\big|\lesssim \sqrt \kappa.
    \end{equation}
\end{enumerate}

\end{lemma}
\begin{proof}
The first statement \eqref{dv-phi-DSS} is a direct consequence of \eqref{eq:self-similarity-outgoing-der} and \eqref{dv r DSS}.

We now write \eqref{wave} as 
\begin{equation}
    \partial_u   \partial_v (r \phi)   = -\frac{m \Omega^2 }{2 r^2} \phi  
\end{equation}
which, after integration and multiplication by $\Omega^{-2}$, gives
\begin{equation*}
 \Omega^{-2} (u,0) \partial_v (r\phi) (u,0) =    \Omega^{-2} (u,0)  \partial_v (r\phi)(p) -  \Omega^{-2} (u,0)  \int_{-1}^u \frac{\Omega^2(u',0)}{-4u} \mu \phi du'.
\end{equation*}
Evaluating the above at $u=-e^{-\Delta}$ and using \eqref{dv-phi-DSS} gives
\begin{equation}
   \Omega^{-2}  (p) \partial_v (r\phi)(p)  = \frac{\Omega^{-2} (-e^{-\Delta}, 0) }{\Omega^{-2} (p) }\Omega^{-2}  (p)  \partial_v (r\phi)(p)  - \Omega^{-2} (-e^{-\Delta}, 0) \int_{-1}^{-e^{-\Delta}}  \frac{\Omega^2(u',0)}{-4 u'} \mu \phi du '.
\end{equation}
Solving for $ \Omega^{-2}   \partial_v (r\phi)(p) $ and using \eqref{eq:Omega-on-cin} yields \eqref{eq:value-of-c2}.
Similarly, we obtain 
\begin{equation}
    \Omega^{-2}   \partial_v (r\phi) (u,0) =  e^{F(u)} \left( \Omega^{-2}    \partial_v (r\phi)(p)   -   \int_{-1}^u \frac{e^{-F(u')}}{-4u'} \mu \phi du' \right) .
\end{equation}
To prove \eqref{dv phi bound} we first note that 
\begin{equation}
|\Omega^{-2}    \partial_v (r\phi)(p) |\lesssim \frac{1}{\kappa \Delta} \sup_{u \in [-1,-e^{-\Delta}]} |\mu \phi| \lesssim \sqrt \kappa,
\end{equation}
where we used \eqref{eq:bounds-on-incoming-cone} and \eqref{dv r bound}. Using the periodicity of $ \Omega^{-2}  \partial_v (r\phi) (u,0) $ and the bounds on $\phi$ and $\mu$ as we did to control $\Omega^{-2}  \partial_v (r\phi) (p)$, we obtain \eqref{dv phi bound} which concludes the proof.  
\end{proof}

\begin{remark}
In \cref{def:discrete-self-similar} we require the first transverse derivative of the geometry and scalar field to be discretely self-similar. Analogously, one can define discretely self-similar profiles of higher order $n\geq 1$. This would impose $n-1$  additional constraints on the higher-order transverse derivatives of the scalar field $\partial_v^{i} \phi|_p$ and on the radius $\partial_v^{i} r|_p$ for $  i \leq n$, analogously to \cref{lem:derivation-of-c1} and \cref{lem:derivation-of-c2}.
\end{remark}

\subsubsection{Estimates along the outgoing cone}
We note that along the outgoing cone, we have that 
\begin{equation}
    \label{eq:gauge-condition-for-omega}
\Omega^2(-1,v)=\frac{1}{1-\kappa}
\end{equation}
which is a consequence of \eqref{eq:Omega-on-vertex} and the gauge condition \eqref{eq:gauge-for-omega-outgoing}.
The self-similarity imposed on the scalar field means  that \[ \partial_v\phi(p)=   \frac{ -e^{\kappa\Delta}}{ (1-\kappa)(1- e^{  \kappa\Delta })} \int_{-1}^{-e^{-\Delta}}    \frac{ e^{-F(u')}}{-4 u'} \mu(u') \phi (u')du '\]
where we recall that $\phi(p) =0$.

\begin{remark}
    Our arguments also apply if we only require the regularity for $\phi$ to be $C^{1,\gamma}_{v}\big([0,\infty)\big),$ where $\gamma=2\kappa/(1-\kappa).$ We notice that this is the same regularity exponent as in \cite{RSR23}. However, while in the examples of \cite{C94} and \cite{RSR23} restricting the regularity of the spacetime was essential in the construction of the exterior solution, in our case we can already achieve this for a smooth scalar field.
\end{remark}

The above choices determine $r$ on $C_0^+$ by the Raychaudhuri equation  \eqref{partial v}  and the constraints at $p$:
\begin{equation*}\begin{dcases}
    \partial_v^2r=-r(\partial_v\phi)^2 \\ r(p)=1,\ \partial_vr(p)= \frac{1}{(1-\kappa)(4-4e^{-\Delta(1+\kappa)})}\int_{-1}^{-e^{-\Delta}}   e^{-F(u')} du'.
\end{dcases}\end{equation*}

We have completed the setup for the characteristic initial data on $C_0^+.$ Next, we prove estimates on $C_0^+$ for all the relevant quantities that will be of use later in the proof. 
\begin{lemma}
  Consider discretely self-similar data as in \cref{def:discrete-self-similar}.    Then, for $0<\kappa\ll 1$ sufficiently small, we have
    \begin{equation}\label{estimates for r initial data}
        \bigg|r(-1,v)-1-\partial_v r (p) v\bigg|  \lesssim \kappa v
    \end{equation}
    and 
    \begin{equation}\label{estimates for dv r initial data}
        \bigg|\partial_vr(-1,v)-\partial_v r (p)\bigg| \lesssim \kappa .
    \end{equation}
\end{lemma}
\begin{proof}
    We denote $v_*=\sup\big\{v\colon \partial_vr>0\text{ on }[0,v]\big\},$ and we prove that $v_*=\infty.$ Since $\partial_v^2r(-1,v)<0,$ we have the bound $0\leq\partial_vr(-1,v)\leq\partial_v r (p) \leq 2$ on $[0,v_*),$ which implies that $1\leq r(-1,v)\leq1+ 2 v$. From \eqref{partial v} get the bound on $[0,v_*)$:
    \begin{equation*}\big|\partial_v^2r(-1,v)\big| \lesssim \kappa\frac{1 + 2v }{(1+v)^{4}}.
    \end{equation*}
    We integrate this bound from 0 to obtain that \eqref{estimates for dv r initial data} holds for all $v\in[0,v_*).$ Taking $\kappa$ sufficiently small, we   obtain that $v_*=\infty$ as desired. Finally, we use \eqref{estimates for dv r initial data} and the fundamental theorem of calculus to get \eqref{estimates for r initial data}.
\end{proof}

Next, we prove the following estimate for $\partial_ur$:
\begin{lemma}
  Consider discretely self-similar data as in \cref{def:discrete-self-similar}. Then, for $0<\kappa\ll 1$ sufficiently small, we have
\begin{equation*}\big|\partial_ur(-1,v)+1\big| \lesssim \kappa \frac{v}{v+1}\end{equation*}
\end{lemma}
\begin{proof}
    We rewrite \eqref{partial uv r} as  \begin{equation*}
        \partial_v\partial_ur^2=-\frac{1}{2}\Omega^2
    \end{equation*}  and integrating along on  $\{u=-1\}$ gives
    \[r\partial_ur+1=-  \frac{v}{4(1-\kappa)},\]
    where we have also used the gauge condition  \eqref{eq:gauge-condition-for-omega}.
    We can rewrite this equation on $\{u=-1\}$ as
    \[r(\partial_ur+1)=r-1-  \frac{v}{4(1-\kappa)}=\bigg( \Omega^{-2}\partial_v r (p)-\frac{1}{4}\bigg)\frac{v}{1-\kappa}+O(\kappa v)=O(\kappa v),\]
where we also used \eqref{dv r bound} and \eqref{estimates for r initial data}.
We then divide by $r$ which concludes the proof.
\end{proof}

The bounds that we proved so far on the initial data imply the smallness of the Hawking mass:
\begin{lemma}\label{lem:hawking-mass-bounded-on-C+}
  Consider discretely self-similar data as in \cref{def:discrete-self-similar}. Then,  \begin{equation*}m(-1,v) \lesssim  \kappa\end{equation*} for all $v\in[0,\infty).$
\end{lemma}
\begin{proof}
    From  $m = \frac 12 \mu r $ and the bound \eqref{dv r bound}  we directly obtain $m(p)\lesssim \kappa$.  
    Plugging  the bounds that we proved above for the initial data, we obtain that
    \[0 \leq \partial_vm \lesssim \frac{\kappa}{(1+v)^{2}}\]
    which shows \cref{lem:hawking-mass-bounded-on-C+} upon integration.
\end{proof}

We also prove the following estimate for $\partial_u\phi:$
\begin{lemma}
    Consider discretely self-similar data as in \cref{def:discrete-self-similar}.   Then, for $0<\kappa\ll 1$ sufficiently small and for  $v\in [0,\infty)$, we have that
    \begin{equation*}\big|\partial_u\phi(-1,v)-\partial_u\phi(p)\big|\lesssim \sqrt{\kappa}\frac{v}{v+1}.\end{equation*}
\end{lemma}
\begin{proof}
    We can rewrite the wave equation \eqref{wave} as $\partial_v(r\partial_u\phi)=-\partial_v\phi\partial_ur.$ Thus, on $u=-1$ we have:
    \begin{equation*}r\partial_u\phi(-1,v)-r\partial_u\phi(p)=O\bigg(\sqrt{\kappa}\frac{v}{v+1}\bigg).\end{equation*}
    We can rewrite this as
    \begin{equation*}r(-1,v)\cdot\big(\partial_u\phi(-1,v)-\partial_u\phi(p)\big)=-\big(r(-1,v)-1\big)\cdot\partial_u\phi(p)+O\big(\sqrt{\kappa}v\big)=O\big(\sqrt{\kappa}v\big),\end{equation*}
    where we also used \eqref{eq:bounds-on-incoming-cone} and \eqref{estimates for r initial data}.
\end{proof}
\subsection{Proof of \texorpdfstring{\cref{prop:data-non-trivial}}{Proposition}}
\label{sec:proof-nontrivial-data}

\begin{proof}[Proof of \cref{prop:data-non-trivial}]
Let $\Delta>0$ and $\kappa>0$ be arbitrary. For the seed data along $C_0^-$ we make the ansatz $\phi = \sqrt{\kappa} \alpha \psi $ for a \textit{fixed} $\psi \in C_c^\infty (-1,e^{-\Delta}), \psi \not \equiv 0$ which we continue self-similarly according to \eqref{eq:self-similar-property-phi}. This ensures \eqref{eq:bounds-on-incoming-cone}. By modulating $\alpha$ as a function of $\Delta$ and $\psi$, we can verify \eqref{eq:definition-F} by the intermediate value theorem. Thus there exists a constant $C_{data}>0$ only depending on $\Delta$ (and the fixed profile $\psi$) satisfying \eqref{eq:bounds-on-incoming-cone}.

 In order to ensure that $\mu$ is self-similar, it suffices to show that      $\Omega^{-2} \partial_v r $ is self-similar. We specify $\Omega^{-2} \partial_v r $ at $p$ according to \eqref{eq:rigidity-of-partial_vr}. Indeed, using \eqref{eq:relation-for-F}, we compute 
\begin{align*}
    \Omega^{-2} \partial_v r & ( e^{-\Delta} u,0)     = \frac{e^{F( u)}}{-u} e^{\Delta +
\kappa \Delta} \left( \Omega^{-2} \partial_v r(p) - \frac{1}{4} \int_{-1}^{-e^{-\Delta} } e^{-F(u')}  du' -  \frac{1}{4} \int_{-e^{-\Delta }}^{ e^{-\Delta}u} e^{-F(u')}  du' \right)\\
    &= \frac{e^{F(u)}}{-u} e^{\Delta +\kappa \Delta} \left( \Omega^{-2} \partial_v r(p) + \Omega^{-2} \partial_v r(p) ( e^{ \Delta (-\kappa -1)    }  -1) -  \frac{e^{-(1+\kappa )\Delta}}{4} \int_{-1}^{u } e^{-F(u')}  du'\right) \\  &=  \Omega^{-2}\partial_v r (u,0)
\end{align*}
which shows  \eqref{eq:self-similarity-for-mu}.

To show \eqref{eq:self-similarity-outgoing-der}, we specify $\Omega^{-2} \partial_v( r \phi) $ at $p$ according to \eqref{eq:value-of-c2}. Then, using \eqref{eq:definition-of-omega-dvrphi} we compute \begin{align*}
    \Omega^{-2}  \partial_v (r\phi) & (e^{-\Delta} u,0)  =  e^{F(u)} e^{\kappa \Delta} \left( \Omega^{-2} \partial_v( r \phi)(p)  -   \int_{-1}^{-e^{-\Delta}} \frac{e^{-F(u')}}{-4u'} \mu \phi du'-   \int_{-e^{-\Delta}}^{e^{-\Delta}u} \frac{e^{-F(u')}}{-4u'} \mu \phi du'\right) \\
    & =  e^{F(u)} e^{\kappa \Delta} \left(  \Omega^{-2} \partial_v( r \phi)(p)  + \frac{ 1- e^{  \kappa\Delta }}{ e^{ \kappa\Delta}}  \Omega^{-2} \partial_v( r \phi)(p)  -   e^{-\kappa\Delta}\int_{-1}^{ u}  \frac{e^{-F(u')}}{-4u'} \mu \phi du'\right) \\ & =  \Omega^{-2}  \partial_v (r\phi) (u,0)
\end{align*}
which shows \eqref{eq:self-similarity-outgoing-der}. Since $|\Omega^{-2} \partial_v( r \phi) (p)|\lesssim\sqrt{\kappa}$, by potentially making $C_{data}$ larger depending only on $\Delta$ and the choice of $\psi$, we readily find data  $\phi = \sqrt \kappa \tilde \psi$ on $C_0^+$ satisfying \eqref{eq:self-similarity-outgoing-der}. The condition \eqref{eq:self-similarity-outgoing-der-2} follows now from \eqref{eq:self-similarity-outgoing-der} and the discrete self-similarity of $\Omega^{-2} \partial_v r$ shown above. Note that the  gauge $\Omega^2 = \frac{1}{1-\kappa}$ in $C_0^+$ uniquely determines $r$ on $C_0^+$. 
\end{proof}

\section{Construction of the exterior region \texorpdfstring{$\mathcal R_{\mathrm{Ext}}$}{Rext}}
\label{sec:exterior-region}

In this section we will prove our main result \cref{main theorem}. Throughout \cref{sec:exterior-region} we fix the parameter $\Delta >0$. We also introduce an additional parameter
\begin{equation}
    0 < \delta \ll 1 
\end{equation}
which is small and independent of all other parameters. For instance, setting $\delta = 10^{-5}$ would work for our proof. 
We consider discretely self-similar characteristic data with period $\Delta$ for ($r,\Omega^2, \phi$) on $C_0^- \cup_p C_0^+$ as defined in \cref{def:discrete-self-similar}. 
According to  \cref{prop:data-non-trivial}, there exists a $C_{data}(\Delta)>0$ such that for all $\kappa>0$, and $N\geq 1$, there exists such nontrivial discretely self-similar characteristic data with period $\Delta$.  
For the rest of \cref{sec:exterior-region} we \textbf{fix} the parameters
\begin{equation} \label{eq:fixed-paramters}
\delta, \Delta, C_{data}(\Delta),  \text{ and } N
\end{equation}
and we will use the following convention.  For non-negative quantities $A$ and $B$, we write $A\lesssim B$ (or also $A=O(B)$) if there exists a constant $\tilde C>0$ only depending on $\delta$, $\Delta$, and $C_{data}(\Delta)$   such that $A\leq \tilde C B$. We will also write  ``for sufficiently small $ A >0   $'' if there exists a constant $A_0 >0$ sufficiently small and only depending on $\delta$, $\Delta$,  and $C_{data}(\Delta)$   such that $A$ satisfies $0< A \leq A_0$. Similarly, we write  ``for sufficiently small $ A >0   $ depending on $B$'' if there exists a constant $A_0>0$ sufficiently small and only depending on $\delta$, $\Delta$, $C_{data}(\Delta)$ and $B$ such that $A$ satisfies $0< A \leq A_0$.

\subsubsection*{Hierarchy of scales} 
In the proofs in \cref{sec:exterior-region} we will introduce further parameters of different scales. Our hierarchy of scales is 
\begin{equation}
   \kappa(\delta,\underline v, D)\ll  \frac{1}{D(\delta, \underline v)}\ll \underline v(\delta) \ll \delta \ll 1,
\end{equation}
where we suppressed the additional dependence on the parameters $\Delta, C_{data}, k$ occurring in \eqref{eq:fixed-paramters}. We will also make use of  $\epsilon \doteq  40 \delta$ and $\eta \doteq 10 \epsilon = 400 \delta$.

By standard local existence results (see e.g.\ \cite[Proposition~1.1]{D05-higgs}),  data as in \cref{def:discrete-self-similar} give rise to a unique solution \begin{equation}\label{eq:solution-by-data}
 (r,\Omega^2,\phi) \in C^{N+2}(\mathcal R_0) \times C^{N+1}(\mathcal R_0) \times C^{N+1}(\mathcal R_0) 
 \end{equation} 
 in the domain
 \[ \mathcal R_0 = \big\{(u,v) \in \mathcal{R}_{\mathrm{Ext}}: v\in[0,v_*(u)]\big\}\]
for some function $v_*(u)>0$.  A priori, however, we have no control on the size of this domain as $u\rightarrow 0$.

 \subsubsection*{The decomposition of $\mathcal R_{\mathrm{Ext}}$ into the regions $\mathcal R_I$--$\mathcal R_{IV}$} 

In order to show \cref{main theorem} we have to show that $(r,\Omega^2,\phi)$ extends to the whole domain $\mathcal R_{\mathrm{Ext}}$. Of course, the difficulty is to show existence as uniqueness follows readily.  To show existence in $\mathcal R_{\mathrm{Ext}}$ we will split up the region $\mathcal R_{\mathrm{Ext}}$ into four regions $\mathcal R_I$--$\mathcal R_{IV}$ and treat each region independently, see \cref{fig:regions}. This decomposition was first established in \cite{RSR23} and also used in \cite{S22}. 
We begin with region $\mathcal R_I$.

\subsection{Region \texorpdfstring{$\mathcal R_I$}{RII}}\label{region I section}
We define 
\[ \mathcal R_I\doteq \bigg\{(u,v)\in \Rext: v\in[0,\underline{v}^{1-\kappa}],\ \frac{v}{|u|^{1-\kappa}}\leq\underline{v}^{1-\kappa}\bigg\},\]
which depends on a small parameter $\underline v>0$.

\begin{proposition}\label{region I main result}
    For  $\underline{v}>0$ sufficiently small and $\kappa>0$ sufficiently small depending on $\underline v$, the solution $(r,\Omega^2,\phi)$ of \eqref{eq:solution-by-data} exists in region $\mathcal R_I$. 
Moreover, $(r,\Omega^2,\phi)$ satisfies the estimates \eqref{A1 improved}--\eqref{A5 improved} in region $\mathcal R_I$.
\end{proposition}
We will prove \cref{region I main result} as a consequence of \cref{region I bootstrap result} below at the end of \cref{region I section}.

\subsubsection*{Bootstrap Assumptions}
Our proof will be based on a bootstrap argument and we will now formulate the bootstrap assumptions. 

We consider the following bootstrap assumptions
\begin{align}\label{A1}
    \big|\Omega^{-2}\partial_v\phi(u,v)-\Omega^{-2}\partial_v\phi(u,0)\big| & \leq 100C\cdot\sqrt{\kappa}\cdot|u|^{-1}\frac{v^{\frac{1-\kappa-2\delta}{1-\kappa}}}{|u|^{1-\kappa-2\delta}} 
\\\label{A2}
    \big|\Omega^{-2}\partial_vr(u,v)-\Omega^{-2}\partial_vr(u,0)\big| &\leq 100C\cdot\kappa\cdot\frac{v^{\frac{1-\kappa-\delta}{1-\kappa}}}{|u|^{1-\kappa-\delta}}
\\\label{A3}
    \big|\partial_u\phi(u,v)-\partial_u\phi(u,0)\big| &\leq 100C\cdot\sqrt{\kappa}\cdot|u|^{-1}\frac{v^{\frac{1-\kappa-\delta}{1-\kappa}}}{|u|^{1-\kappa-\delta}}
\\\label{A4}
    \big|\partial_ur(u,v)-\partial_ur(u,0)\big| &\leq 100C\cdot\kappa\cdot\frac{v^{\frac{1-\kappa-\delta}{1-\kappa}}}{|u|^{1-\kappa-\delta}}
\\\label{A5}
    \bigg|\log\frac{1-\kappa}{|u|^{\kappa}}\Omega^2-\log\Big(|u|^{-\kappa}e^{-F(u)}\Big)\bigg| &\leq 100C\cdot\kappa\cdot\frac{v^{\frac{1-\kappa-\delta}{1-\kappa}}}{|u|^{1-\kappa-\delta}}
\end{align}
for some constant $C >0 $ depending only on the fixed parameters in \eqref{eq:fixed-paramters}. In particular, in view of our convention, we have $C \lesssim 1$.

\begin{proposition}\label{region I bootstrap result}
    For any $(u_1,v_1)\in \mathcal R_I,$ we assume that  the bootstrap assumptions \eqref{A1}--\eqref{A5} hold for all $(u,v)\in\mathcal{R}^I_{u_1,v_1}=[-1,u_1]\times[0,v_1].$ Then, the following improved estimates hold in $\mathcal{R}^I_{u_1,v_1}$:
    \begin{align}\label{A1 improved}
    \big|\Omega^{-2}\partial_v\phi(u,v)-\Omega^{-2}\partial_v\phi(u,0)\big| & \leq 10C\cdot\sqrt{\kappa}\cdot|u|^{-1}\frac{v^{\frac{1-\kappa-2\delta}{1-\kappa}}}{|u|^{1-\kappa-2\delta}}
\\ \label{A2 improved}
    \big|\Omega^{-2}\partial_vr(u,v)-\Omega^{-2}\partial_vr(u,0)\big|& \leq 10C\cdot\kappa\cdot\frac{v^{\frac{1-\kappa-\delta}{1-\kappa}}}{|u|^{1-\kappa-\delta}}
\\\label{A3 improved}
    \big|\partial_u\phi(u,v)-\partial_u\phi(u,0)\big|& \leq 10C\cdot\sqrt{\kappa}\cdot|u|^{-1}\frac{v^{\frac{1-\kappa-\delta}{1-\kappa}}}{|u|^{1-\kappa-\delta}}
\\ \label{A4 improved}
    \big|\partial_ur(u,v)-\partial_ur(u,0)\big| & \leq 10C\cdot\kappa\cdot\frac{v^{\frac{1-\kappa-\delta}{1-\kappa}}}{|u|^{1-\kappa-\delta}}
\\\label{A5 improved}
    \bigg|\log\frac{1-\kappa}{|u|^{\kappa}}\Omega^2-\log\Big(|u|^{-\kappa}e^{-F(u)}\Big)\bigg|& \leq 10C\cdot\kappa\cdot\frac{v^{\frac{1-\kappa-\delta}{1-\kappa}}}{|u|^{1-\kappa-\delta}}
\end{align}
\end{proposition}
\begin{proof}[Proof of \cref{region I bootstrap result}]
The bootstrap assumptions hold initially by continuity: We first notice that the left hand sides vanish on the incoming cone $v=0.$ Similarly, by choosing $C$ sufficiently large, the bootstrap assumptions hold on $u=-1$ with $C$ instead of $100C,$ because of the choice of initial data and the previous estimates on the outgoing cone $\{u=-1\}.$
For the rest of the proof of   \cref{region I bootstrap result},  we will assume that the bootstrap assumptions hold in $\mathcal{R}^I_{u_1,v_1}$ and we prove the desired improved estimates.

\subsubsection*{Preliminary Estimates}
As a consequence of the bootstrap assumptions and the estimates on the initial data and by choosing $\underline v>0$ sufficiently small, we get the bounds:

\begin{equation}\label{B1}
    \big|\Omega^{-2}\partial_v\phi(u,v)\big|\leq 2C\cdot\sqrt{\kappa}\cdot|u|^{-1}
\end{equation}
\begin{equation}\label{B2}
    0 <  \Omega^{-2}\partial_vr(u,v) \leq1,\ \big|\Omega^{-2}\partial_vr(u,v)-\Omega^{-2}\partial_vr(u,0)\big|\leq \kappa
\end{equation}
\begin{equation}\label{B3}
    \big|\partial_u\phi(u,v)\big|\leq 2C\cdot\sqrt{\kappa}\cdot|u|^{-1}
\end{equation}
\begin{equation}\label{B4}
    \big|\partial_ur(u,v)\big|\leq2,\ \big|\partial_ur(u,v)-\partial_ur(u,0)\big|\leq\kappa
\end{equation}
\begin{equation}\label{B5}
    1-2C\kappa\leq\frac{1-\kappa}{|u|^{\kappa}}\Omega^2\leq 1+2C\kappa
\end{equation}

We also have the estimates for $r$:

\begin{lemma}\label{lem:radius-bounds}
     The area radius $r$ satisfies the following bounds:
     \[0\leq r(u,v)+u\leq4Cv|u|^{\kappa}\leq|u|\]
     \[|r^{-1}(u,v)+u^{-1}|\leq4C\cdot\frac{1}{|u|}\cdot\frac{v}{|u|^{1-\kappa}}\]
     \[r(u,v)=-u\cdot\bigg(1+\frac{1}{4}\cdot\frac{1}{1-\kappa}\cdot\frac{v}{|u|^{1-\kappa}}+O\bigg(\kappa\cdot\frac{v}{|u|^{1-\kappa}}\bigg)\bigg)\]
\end{lemma}
\begin{proof}
     Since $\partial_v r>0$ and $r=-u$ on $C_0^-$ we have that $-u\leq r(u,v)$ in $\mathcal R_{I}$ Moreover, we have
\[r(u,v)+u\leq\int_0^v\Omega^2\cdot\Omega^{-2}\partial_vr(u,v')dv'\leq4Cv|u|^{\kappa}\leq|u|\]
for $\underline v$ sufficiently small. Moreover,
\[|r^{-1}(u,v)+u^{-1}|\leq \frac{|r(u,v)+u|}{u^2} \leq4C\cdot\frac{1}{|u|}\cdot\frac{v}{|u|^{1-\kappa}}\]
as well as 
\[r(u,v)=-u+\int_0^v\Omega^2\cdot\Omega^{-2}\partial_vr(u,v')dv'=-u+\int_0^v\frac{|u|^{\kappa}}{1-\kappa}\cdot\big(1+O(\kappa)\big)\cdot\bigg(\frac{1}{4}+O(\kappa)\bigg)dv'\]
which concludes the proof of \cref{lem:radius-bounds}.
\end{proof}

\subsubsection*{Main Estimates}
We will now improve the bootstrap assumptions, completing the proof of   \cref{region I bootstrap result}. We introduce the following notation: for any function $\psi$  we define:
\[\widetilde{\psi}(u,v):=\psi(u,v)-\psi(u,0)\]

We begin by improving assumption \eqref{A5}. Notice that we have the equation:
\[\partial_u\partial_v\log\frac{1-\kappa}{|u|^{\kappa}}\Omega^2=-2\partial_u\phi\partial_v\phi+\frac{\Omega^2}{2r^2}+\frac{2\partial_vr\partial_ur}{r^2}=-2\partial_u\phi\partial_v\phi+\frac{\Omega^2}{2r^2}\bigg(1+4\Omega^{-2}\partial_vr\partial_ur\bigg)=\]\[=-2\partial_u\phi\partial_v\phi+\frac{\Omega^2}{2r^2}\bigg(4\widetilde{\Omega^{-2}\partial_vr}\cdot\partial_ur+4\Omega^{-2}\partial_vr(u,0)\cdot\widetilde{\partial_ur}+1-4\Omega^{-2}\partial_vr(u,0)\bigg)\]
Using the bootstrap assumptions and the preliminary estimates, we have
\begin{equation}\partial_u\partial_v\log\frac{1-\kappa}{|u|^{\kappa}}\Omega^2=O\big(\kappa\cdot|u|^{-2+\kappa}\big)
\label{eq:estimate-on-Omega-wave}
\end{equation}
We integrate from $-1$ to $u,$ and use the fact that $(1-\kappa)\Omega^2(-1,v)=1$, in order to obtain that:
\[\bigg|\partial_v\log\frac{1-\kappa}{|u|^{\kappa}}\Omega^2\bigg|=O\big(\kappa\cdot|u|^{-1+\kappa}\big)\]
We integrate from $0$ to $v$ and obtain.
\[\bigg|\log\frac{1-\kappa}{|u|^{\kappa}}\Omega^2-\log\Big(|u|^{-\kappa}e^{-F(u)}\Big)\bigg|\lesssim\kappa\cdot\frac{v}{|u|^{1-\kappa}}\lesssim\kappa\cdot\frac{v^{\frac{1-\kappa-\delta}{1-\kappa}}}{|u|^{1-\kappa-\delta}}\cdot\frac{v^{\frac{\delta}{1-\kappa}}}{|u|^{\delta}} \leq \kappa\cdot\frac{v^{\frac{1-\kappa-\delta}{1-\kappa}}}{|u|^{1-\kappa-\delta}}\cdot\underline v^{\delta}. \]
By choosing $\underline v$ sufficiently small we improve the constant in  \eqref{A5}. Thus, we proved \eqref{A5 improved}.

We improve assumption \eqref{A2} next. We have by equation \eqref{partial v}, together with the bootstrap assumptions and the preliminary estimates:
\[\big|\widetilde{\Omega^{-2}\partial_vr}\big|\leq\int_0^vr\cdot\Omega^2\cdot\big(\Omega^{-2}\partial_v\phi\big)^2(u,v')dv'\lesssim\int_0^v\kappa\cdot|u|^{\kappa-1}dv'\lesssim\kappa\cdot\frac{v}{|u|^{1-\kappa}},\]
which improves assumption \eqref{A2} as before.
Thus, we proved \eqref{A2 improved}.

We improve assumption \eqref{A3}. From the wave equation \eqref{wave}, the bootstrap assumptions, and the preliminary estimates we have that
\[\big|\widetilde{r\partial_u\phi}\big|(u,v)\leq\int_0^v\big|\partial_ur\big|\cdot\Omega^2\cdot\big|\Omega^{-2}\partial_v\phi \big|dv'\lesssim\sqrt{\kappa}\cdot\frac{v}{|u|^{1-\kappa}}.\]
Moreover, we also have the estimate:
\[r(u,v)\cdot\big|\widetilde{\partial_u\phi}\big|\leq\big|\widetilde{r\partial_u\phi}\big|(u,v)+\big|\widetilde{r}\big|(u,v)\cdot\big|\partial_u\phi\big|(u,0)\lesssim\sqrt{\kappa}\cdot\frac{v}{|u|^{1-\kappa}},\]
where we used the previous bound for $\widetilde{r\partial_u\phi},$ the bound on initial data for $\partial_u\phi(u,0)$, and the preliminary bound for $\widetilde{r}(u,v).$ This improves assumption \eqref{A3} as before, so we obtain \eqref{A3 improved}.

We improve assumption \eqref{A4}. Note that we can write equation \eqref{partial uv r} as:
\[\partial_v\big(r\cdot\widetilde{\partial_ur}\big)=-\frac{\Omega^2}{4}+\partial_vr=\Omega^2\cdot\bigg(\widetilde{\Omega^{-2}\partial_vr}(u,v)+\Omega^{-2}\partial_vr(u,0)-\frac{1}{4}\bigg)=O\big(\kappa\cdot|u|^{\kappa}\big)\]
We integrate this and divide by $r(u,v)$ in order to get:
\[\widetilde{\partial_ur}=O\bigg(\kappa\cdot\frac{v}{|u|^{1-\kappa}}\bigg)\]

Finally, we improve assumption \eqref{A1}. We use \eqref{wave} in order to obtain a suitable equation for $\widetilde{\Omega^{-2}\partial_v\phi}.$ We write \eqref{wave} as
\begin{equation*}\partial_u\big(\Omega^{-2}\partial_v\phi\big)=-\Omega^{-2}\partial_v\phi\cdot\bigg(\frac{\partial_ur}{r}-\partial_u\log\Omega^{-2}\bigg)-\frac{1}{r}\partial_u\phi\cdot\Omega^{-2}\partial_v r\end{equation*}
from which we derive the following equation for  $\widetilde{\Omega^{-2}\partial_v\phi}$:
\begin{multline*}
\partial_u \left( \widetilde{\Omega^{-2}\partial_v\phi}\right)+ \widetilde{\Omega^{-2}\partial_v\phi}\cdot\big(\partial_u\log r(u,0)-F'(u)\big) \\ 
=-\Omega^{-2}\partial_v\phi\cdot\bigg(\frac{1}{r}\cdot\widetilde{\partial_ur}-\widetilde{r^{-1}}-\widetilde{\partial_u\log\Omega^{-2}}\bigg)-\frac{1}{r}\cdot\Omega^{-2}\partial_vr\cdot\widetilde{\partial_u\phi} \\ -\partial_u\phi(u,0)\cdot\frac{1}{r}\cdot\widetilde{\Omega^{-2}\partial_vr}-\partial_u\phi(u,0)\cdot\Omega^{-2}\partial_vr(u,0)\cdot\widetilde{r^{-1}}   
\end{multline*}
Using \eqref{eq:estimate-on-Omega-wave} we estimate 
\begin{equation}
|    \widetilde{\partial_u\log\Omega^{-2}}|\lesssim \kappa |u|^{-2+\kappa} v 
\end{equation} which together with the bootstrap assumptions and the preliminary estimates, gives us
\[\partial_u \left( \widetilde{\Omega^{-2}\partial_v\phi}\right)+ \widetilde{\Omega^{-2}\partial_v\phi}\cdot\big(\partial_u\log r(u,0)-F'(u)\big)=O\bigg(\sqrt{\kappa}\cdot|u|^{-2}\frac{v^{\frac{1-\kappa-\delta}{1-\kappa}}}{|u|^{1-\kappa-\delta}}\bigg)\]
We integrate this equation to obtain:
\begin{align*}|\widetilde{\Omega^{-2}\partial_v\phi}|(u,v) & \lesssim\frac{e^{F(u)}}{|u|}\cdot|\widetilde{\Omega^{-2}\partial_v\phi}|(-1,v)+\frac{e^{F(u)}}{|u|}\cdot\int_{-1}^u\frac{|u'|}{e^{F(u')}}\cdot\sqrt{\kappa}\cdot|u'|^{-2}\cdot\frac{v^{\frac{1-\kappa-\delta}{1-\kappa}}}{|u'|^{1-\kappa-\delta}}du' \\ & \lesssim\sqrt{\kappa}\cdot\frac{v}{|u|^{1+\kappa}}+\frac{\sqrt{\kappa}}{|u|^{1+\kappa}}\cdot\int_{-1}^u|u'|^{-1+\kappa}\cdot\frac{v^{\frac{1-\kappa-\delta}{1-\kappa}}}{|u'|^{1-\kappa-\delta}}du'\\ & \lesssim\sqrt{\kappa}\cdot\frac{v}{|u|^{1+\kappa}}+\sqrt{\kappa}\cdot|u|^{-1}\cdot\frac{v^{\frac{1-\kappa-\delta}{1-\kappa}}}{|u|^{1-\kappa-\delta}}
\end{align*}
where we used the initial data estimates \eqref{ansatz for phi C0+} and \eqref{eq:bound-on-F}. Since $1\leq|u|^{-1}$ and $\underline{v}^{\delta}\ll1,$ we proved \eqref{A1 improved}. This completes the proof of  \cref{region I bootstrap result}.
\end{proof}

\begin{proof}[Proof of \cref{region I main result}]
     \cref{region I bootstrap result} implies  \cref{region I main result} by a standard local existence argument.
\end{proof}
In particular, the solution $(r,\Omega^2,\phi)$ exists on the spacelike hypersurface $\big\{v=\underline{v}^{1-\kappa}|u|^{1-\kappa}\big\},$ and  satisfies the following smallness result:
\begin{corollary}\label{region I corollary}
The solution $(r,\Omega^2,\phi)$ satisfies the following estimates on the spacelike hypersurface $\big\{v=\underline{v}^{1-\kappa}|u|^{1-\kappa}\big\}$:
{\allowdisplaybreaks  \begin{align*}\big|\Omega^{-2}\partial_v\phi\big| &\leq 2C\cdot\sqrt{\kappa}\cdot|u|^{-1} \\
  \big|\partial_u\phi\big|&\leq 2C\cdot\sqrt{\kappa}\cdot|u|^{-1}\\
\big|4\Omega^{-2}\partial_vr-1\big|&\leq 2C\cdot\kappa,\\ \big|\partial_ur+1\big|&\leq 2C\cdot\kappa,\\ \big||u|^{-\kappa}\Omega^2-1\big|&\leq 2C\cdot\kappa\\ 
r&=-u\cdot\left(1+\frac{1}{4}\cdot\underline{v}+O(\kappa^{1-\delta})\right)
\end{align*}}
for some constant $C>0$ with $C\lesssim 1$. 
\end{corollary}

\subsection{Region \texorpdfstring{$\mathcal R_{II}$}{RII}}\label{region II section}
In this section we prove the existence of the solution in region $\mathcal R_{II}$. As explained in the introduction, we consider a new system of double null coordinates $(u,V)$, with respect to which the solution is close to Minkowski space in standard null coordinates. We define on $\mathcal R_{\mathrm{Ext}} \setminus \mathcal R_{I}$  the new coordinates
\[V(u,v)=v^{\frac{1}{1-\kappa}}, \quad u(u,v) = u, \]
in mild abuse of notation. 
We remark that the associated transition map is a smooth diffeomorphism on $\mathcal R_{\mathrm{Ext}} \setminus \mathcal R_{I}$ and the metric is given by
\[g=-\bar{\Omega}^2(u,V) dudV+r^2d\sigma_{S^2},\]
where the new lapse function $\bar \Omega^2$ is given as 
\[\bar{\Omega}^2(u,V) =(1-\kappa)\cdot V^{-\kappa} \cdot\Omega^2(u,v(V)). \]
We define 
\[ \mathcal R_{II} =\bigg\{(u,V):\ u\in[-1,0),\ V\in[0,\underline{v}^{-1}],\ \underline{v}\leq\frac{V}{|u|}\leq\underline{v}^{-1}\bigg\}\]
and the main result of this section is
\begin{proposition}\label{region II main result}
 For  $\underline{v}>0$ sufficiently small and $\kappa>0$ sufficiently small depending on $\underline v$, the solution $(r,\Omega^2,\phi)$ exists in region $\mathcal R_{II}$. 
Moreover, the solution satisfies the estimates \eqref{A1 II improved}--\eqref{A5 II improved} in region $\mathcal R_{II}$.
\end{proposition}
The proof follows the ideas of \cite{RSR23}, by proving estimates which correspond to applying the Gronwall lemma starting from data at $\big\{V=\underline{v}|u|\big\}\cup\big\{u=-1,\ V\in[\underline{v},\underline{v}^{-1}]\big\}.$

\subsubsection*{Bootstrap Assumptions}

We introduce a large constant $D\gg1$ which depends on $\underline{v}$ (and on the parameters in \eqref{eq:fixed-paramters}). We will choose  $\kappa$ sufficiently small depending on $D, \underline{v}$ (and on the parameters in \eqref{eq:fixed-paramters})  such that
\begin{equation}\label{bound for D}
    \kappa^{\delta}\cdot\exp\bigg(\bigg(\frac{D}{\underline{v}}\bigg)^{100}\bigg)\ll1 . 
\end{equation}
We make the following bootstrap assumptions in region $\mathcal R_{II}$:
\begin{align}\label{A1 II}
    \bigg|\frac{V}{4}-u\bigg|\cdot|\partial_V\phi|\cdot\exp\bigg(D\cdot\frac{V}{u}\bigg)&\leq 100\kappa^{\frac{1}{2}-\delta}
\\ \label{A2 II}
    \bigg|\frac{V}{4}-u\bigg|\cdot|\partial_u\phi|\cdot\exp\bigg(D\cdot\frac{V}{u}\bigg)&\leq 100\kappa^{\frac{1}{2}-\delta}
\\ 
\label{A3 II}
    \bigg|\partial_Vr-\frac{1}{4}\bigg|\cdot\exp\bigg(D\cdot\frac{V}{u}\bigg)&\leq 100\kappa^{1-10\delta}
\\
\label{A4 II}
    \big|\partial_ur+1\big|\cdot\exp\bigg(D\cdot\frac{V}{u}\bigg)&\leq 100\kappa^{1-10\delta}
\\
\label{A5 II}
    \big|\bar{\Omega}^2-1\big|\cdot\exp\bigg(D\cdot\frac{V}{u}\bigg)&\leq 100\kappa^{1-8\delta}
\end{align}

The local existence result implies that  \cref{region II main result} follows once we prove 
\begin{proposition}\label{region II bootstrap result}
For any $(u_1,V_1)\in \mathcal R_{II},$ we assume that the solution satisfies the bootstrap assumptions \eqref{A1 II}--\eqref{A5 II} for all $(u,V)\in\mathcal{R}^{II}_{u_1,V_1}=\mathcal R_{II}\cap\big([-1,u_1]\times[0,V_1]\big).$ Then the solution satisfies the improved estimates in $\mathcal{R}^{II}_{u_1,V_1}:$
    \begin{align}\label{A1 II improved}
    \bigg|\frac{V}{4}-u\bigg|\cdot|\partial_V\phi|\cdot\exp\bigg(D\cdot\frac{V}{u}\bigg)&\leq 10\kappa^{\frac{1}{2}-\delta}\\
    \label{A2 II improved}
    \bigg|\frac{V}{4}-u\bigg|\cdot|\partial_u\phi|\cdot\exp\bigg(D\cdot\frac{V}{u}\bigg)&\leq 10 \kappa^{\frac 12 - \delta} \\
    \label{A3 II improved}
    \bigg|\partial_Vr-\frac{1}{4}\bigg|\cdot\exp\bigg(D\cdot\frac{V}{u}\bigg)&\leq 10\kappa^{1-10\delta}
\\ \label{A4 II improved}
\big|\partial_ur+1\big|\cdot\exp\bigg(D\cdot\frac{V}{u}\bigg)&\leq 10\kappa^{1-10\delta}
\\ \label{A5 II improved}
    \big|\bar{\Omega}^2-1\big|\cdot\exp\bigg(D\cdot\frac{V}{u}\bigg)&\leq 10\kappa^{1-8\delta}
\end{align}
\end{proposition}

We notice that the bootstrap assumptions hold initially on $\big\{V=\underline{v}|u|\big\}$ since we can rewrite the estimates in  \cref{region I corollary} in $(u,V)$ coordinates as:
\[\big|\partial_V\phi\big|(u,-\underline{v}u)\leq 4C\cdot\sqrt{\kappa}\cdot|u|^{-1},\ \big|\partial_u\phi\big|(u,-\underline{v}u)\leq 2C\cdot\sqrt{\kappa}\cdot|u|^{-1}\]
\[\big|4\partial_Vr(u,-\underline{v}u)-1\big|\leq\kappa^{1-2\delta},\ \big|\partial_ur(u,-\underline{v}u)+1\big|\leq 2C\cdot\kappa\]
\[1-C\kappa^{1-\delta}\leq\bar{\Omega}^2(u,-\underline{v}u)\leq1+C\kappa^{1-\delta}\]
Similarly, the bootstrap assumptions hold initially on $\big\{u=-1,\ V\in[\underline{v},\underline{v}^{-1}]\big\},$ since we can rewrite the estimates on the outgoing initial data in $(u,V)$ coordinates as:
\[\big|\partial_V\phi\big|(-1,V)\leq\kappa^{\frac{1}{2}-\frac{\delta}{2}}\cdot V^{-1},\ \big|\partial_u\phi\big|(-1,V)\leq \kappa^{\frac{1}{2}-\frac{\delta}{2}}\cdot V^{-1}\]
\[\big|4\partial_Vr(-1,V)-1\big|\leq\kappa^{1-2\delta},\ \big|\partial_ur(-1,V)+1\big|\leq \kappa^{1-2\delta}\]
\[1-\kappa^{1-\delta}\leq\bar{\Omega}^2(-1,V)\leq1+\kappa^{1-\delta}\]

For the rest of the section we prove   \cref{region II bootstrap result}, so we assume that the bootstrap assumptions hold in $\mathcal{R}^{II}_{u_1,V_1}$ and we prove the desired improved estimates.

\subsubsection*{Preliminary Estimates}
As a consequence of the bootstrap assumptions, we have the following simple bounds
\begin{align*} |\partial_u\phi| & \leq\kappa^{\frac{1}{2}-2\delta}|u|^{-1}, \quad  |\partial_V\phi|\leq\kappa^{\frac{1}{2}-2\delta}|u|^{-1}  \\ \frac{1}{8} & \leq|\partial_Vr|\leq\frac{1}{2}\leq|\partial_ur|, \quad |\bar{\Omega}^2| \leq2. 
\end{align*}
We also have the following estimate for $r:$
\begin{lemma}
    The area radius function $r$ satisfies:
    \[r(u,V)=\frac{V}{4}-u+O(|u|\kappa^{1-11\delta})\]
\end{lemma}
\begin{proof}
    Using the bootstrap assumptions and  \cref{region I corollary}, we have
    \[r(u,V)=r(u,-\underline{v}u)+\int_{-\underline{v}u}^V\partial_Vr(u,V')dV'=\frac{V}{4}-u+O(|u|\kappa^{1-\delta})+O\bigg(\underline{v}^{-1}\cdot\exp\bigg(\frac{D}{\underline{v}}\bigg)\cdot|u|\kappa^{1-10\delta}\bigg)\]
    which leads to the desired conclusion by \eqref{bound for D}.
\end{proof}

\subsubsection*{Main Estimates}
We improve the bootstrap assumptions, completing the proof of  \cref{region II bootstrap result}. We begin by improving assumption \eqref{A1 II}. We can rewrite the wave equation \eqref{wave} as:
\[\partial_u\bigg(\exp\bigg(D\cdot\frac{V}{u}\bigg)r\partial_V\phi\bigg)=-D\cdot\frac{V}{u^2}\cdot\exp\bigg(D\cdot\frac{V}{u}\bigg)r\partial_V\phi-\exp\bigg(D\cdot\frac{V}{u}\bigg)\partial_Vr\partial_u\phi\]
We multiply the equation in order to obtain:
\[\partial_u\bigg|\exp\bigg(D\cdot\frac{V}{u}\bigg)r\partial_V\phi\bigg|^2+2D\cdot\frac{V}{u^2}\cdot\bigg|\exp\bigg(D\cdot\frac{V}{u}\bigg)r\partial_V\phi\bigg|^2=-2\exp\bigg(2D\cdot\frac{V}{u}\bigg)\partial_Vr\partial_u\phi\cdot r\partial_V\phi\]
We denote $b(V)=\max\big\{-1,-V/\underline{v}\big\}.$ We integrate the above equation and absorb an error term to obtain the energy estimate
\begin{align*} \bigg|\exp\bigg(D & \cdot\frac{V}{u}\bigg)r\partial_V\phi\bigg|^2(u,V)  +D\int_{b(V)}^u\frac{V}{u'^2}\cdot\bigg|\exp\bigg(D\cdot\frac{V}{u'}\bigg)r\partial_V\phi\bigg|^2(u',V)du' \\   & \leq\big|\exp\big(-D\underline{v}\big)r\partial_V\phi\big|^2(b(V),V)+D^{-1}\int_{b(V)}^u\frac{u'^2}{V}\cdot\bigg|\exp\bigg(D\cdot\frac{V}{u'}\bigg)\partial_Vr\partial_u\phi\bigg|^2(u',V)du'. 
\end{align*} 
Using the bootstrap assumptions, the preliminary estimates, and the bounds of the solution on $\big\{V=-\underline{v}u\big\}$ and $\big\{u=-1,\ V\in[\underline{v},\underline{v}^{-1}]\big\}$, we get 
\[\bigg|\exp\bigg(D\cdot\frac{V}{u}\bigg)r\partial_V\phi\bigg|^2(u,V)\leq \kappa^{1-2\delta}+D^{-1}\cdot\kappa^{1-2\delta}\cdot O(\underline{v}^{-1}).\]
Using the preliminary estimates for $r$, we obtain that
\[\bigg|\frac{V}{4}-u\bigg|\cdot|\partial_V\phi|\cdot\exp\bigg(D\cdot\frac{V}{u}\bigg)(u,V)\leq\kappa^{\frac{1}{2}-\delta}\cdot\bigg(1+D^{-\frac{1}{2}}\cdot O(\underline{v}^{-\frac{1}{2}})+O(\kappa^{1-12\delta})\bigg)\]
which improves the bootstrap assumption \eqref{A1 II} for $D$ large enough. Thus, we proved \eqref{A1 II improved}.

Similarly, we improve assumption \eqref{A2 II}. We can rewrite the wave equation \eqref{wave} as
\[\partial_V\bigg(\exp\bigg(D\cdot\frac{V}{u}\bigg)r\partial_u\phi\bigg)=D\cdot\frac{1}{u}\cdot\exp\bigg(D\cdot\frac{V}{u}\bigg)r\partial_u\phi-\exp\bigg(D\cdot\frac{V}{u}\bigg)\partial_ur\partial_V\phi\]
We multiply the equation in order to obtain:
\[\partial_V\bigg|\exp\bigg(D\cdot\frac{V}{u}\bigg)r\partial_u\phi\bigg|^2+2D\cdot\frac{1}{|u|}\cdot\bigg|\exp\bigg(D\cdot\frac{V}{u}\bigg)r\partial_u\phi\bigg|^2=-2\exp\bigg(2D\cdot\frac{V}{u}\bigg)\partial_ur\partial_V\phi\cdot r\partial_u\phi\]
We integrate this equation and absorb an error term to obtain the energy estimate:
\[\bigg|\exp\bigg(D\cdot\frac{V}{u}\bigg)r\partial_u\phi\bigg|^2(u,V)+D\int_{-\underline{v}u}^V\frac{1}{|u|}\cdot\bigg|\exp\bigg(D\cdot\frac{V'}{u}\bigg)r\partial_u\phi\bigg|^2(u,V')dV'\leq\]\[\leq\big|\exp\big(-D\underline{v}\big)r\partial_u\phi\big|^2(u,-\underline{v}u)+D^{-1}\int_{-\underline{v}u}^V|u|\cdot\bigg|\exp\bigg(D\cdot\frac{V'}{u}\bigg)\partial_ur\partial_V\phi\bigg|^2(u,V')dV'\]
Using the bootstrap assumptions, the preliminary estimates, and the bounds of the solution on $\{V=-\underline{v}u\}$, we get that for some constant $C'>0$ independent of $\kappa,D,\underline{v},$ and $\delta$:
\[\bigg|\exp\bigg(D\cdot\frac{V}{u}\bigg)r\partial_u\phi\bigg|^2(u,V)\leq C'\kappa+D^{-1}\cdot\kappa^{1-2\delta}\cdot O(\underline{v}^{-1})\]
Using the preliminary estimates for $r$, we obtain that:
\[\bigg|\frac{V}{4}-u\bigg|\cdot|\partial_u\phi|\cdot\exp\bigg(D\cdot\frac{V}{u}\bigg)(u,V)\leq\kappa^{\frac{1}{2}-\delta}\cdot\bigg(1+D^{-\frac{1}{2}}\cdot O(\underline{v}^{-\frac{1}{2}})+O(\kappa^{1-12\delta})\bigg),\]
which improves the bootstrap assumption \eqref{A2 II} and proves \eqref{A2 II improved}.

Next, we improve assumption \eqref{A3 II}.  We can rewrite the Raychaudhuri equation \eqref{partial v} as:
\[\partial_V\bigg(\bar{\Omega}^{-2}\partial_Vr-\frac{1}{4}\bigg)=-r\bar{\Omega}^{-2}\big(\partial_V\phi\big)^2\]
We integrate this equation to obtain:
\[\bigg|\bar{\Omega}^{-2}\partial_Vr-\frac{1}{4}\bigg|(u,V)\leq\bigg|\bar{\Omega}^{-2}\partial_Vr-\frac{1}{4}\bigg|(u,-\underline{v}u)+\int_{-\underline{v}u}^Vr\bar{\Omega}^{-2}\big(\partial_V\phi\big)^2(u,V')dV'\]
Using the bootstrap assumptions, the preliminary estimates, and the bounds of the solution on $\{V=-\underline{v}u\}$, we get:
\[\bigg|\exp\bigg(D\cdot\frac{V}{u}\bigg)\cdot\bigg(\bar{\Omega}^{-2}\partial_Vr-\frac{1}{4}\bigg)\bigg|(u,V)\leq\kappa^{1-3\delta}+\kappa^{1-2\delta}\cdot\exp\bigg(3\frac{D}{\underline{v}}\bigg)\cdot O(\underline{v}^{-1})\leq\kappa^{1-4\delta}\]
As a result, we obtain that:
\[\bigg|\exp\bigg(D\cdot\frac{V}{u}\bigg)\cdot\bigg(\partial_Vr-\frac{1}{4}\bigg)\bigg|\leq\bigg|\bar{\Omega}^{2}\cdot\exp\bigg(D\cdot\frac{V}{u}\bigg)\cdot\bigg(\bar{\Omega}^{-2}\partial_Vr-\frac{1}{4}\bigg)\bigg|+\frac{1}{4}\bigg|\exp\bigg(D\cdot\frac{V}{u}\bigg)\cdot\big(\bar{\Omega}^{2}-1\big)\bigg|\leq 2\kappa^{1-10\delta}\]
We point out that since we use the bootstrap assumption \eqref{A5 II} in the above inequality, the assumption \eqref{A3 II} must have a lower power of $\kappa$ than \eqref{A5 II}. Also, we remark that a completely analogous proof improves assumption \eqref{A4 II}. Thus, we proved \eqref{A3 II improved} and \eqref{A4 II improved}.

Finally, we improve assumption \eqref{A5 II}. The key point is that we avoid using the estimates \eqref{A3 II} and \eqref{A4 II} (which do not have small enough powers of $\kappa$ for our purposes) by using the Hawking mass instead. We use that $\bar \Omega^{-2} (-\partial_u r) = \frac{1-\mu}{4 \partial_V r}$ and substitute it in \eqref{eq:hawking-mass-v} to obtain 
\begin{equation}\label{dv m} \partial_Vm=\big(2\partial_Vr\big)^{-1}\cdot\bigg(1-\frac{2m}{r}\bigg)\cdot r^2(\partial_V\phi)^2.
\end{equation}
Using the bootstrap assumptions and the preliminary estimates, we get that $|\partial_Vm|\leq\kappa^{1-3\delta}.$ Also, by the bounds on the hypersurface $\{V=-\underline{v}u\},$ we have $\frac{m}{r}(u,-\underline{v}u)=O(\kappa).$ Therefore, we have that in region $\mathcal{R}_{II}$:
\[\frac{m}{r}\leq\kappa^{1-4\delta}\]

The relation between \eqref{A5 II} and the Hawking mass follows from the equation
\begin{equation}\label{du dv Omega using m}
\partial_u\partial_V\log\bar{\Omega}^2=-2\partial_u\phi\partial_V\phi+ \frac{\bar{\Omega}^2}{r^2}\cdot\frac{m}{r}.
\end{equation}
Using the bootstrap assumptions and the bound for the Hawking mass, we obtain that: \[\big|\partial_u\partial_V\log\bar{\Omega}^2\big|\leq|u|^{-2}\kappa^{1-5\delta}\]
We claim that the analysis in region $\mathcal R_{I}$ implies that $\partial_u\log\Omega^2\big(u,|\underline{v}u|^{1-\kappa}\big)=O(\kappa|u|^{-1}).$ Indeed, we already proved that $\widetilde{\partial_u\log\Omega^2}=O(\kappa|u|^{-1})$ and $\partial_u\log\Omega^2(u,0)=O(\kappa|u|^{-1}).$ Moreover, we notice that we also have $\partial_u\log\Omega^2\big(u,|\underline{v}u|^{1-\kappa}\big)=\partial_u\log\bar{\Omega}^2(u,-\underline{v}u).$ Thus, we integrate the above in the $V$ direction in order to get:
\begin{equation}\big|\partial_u\log\bar{\Omega}^2\big|\leq V^{-1}\kappa^{1-6\delta}.\label{eq:estimate-partialulogomega-region-II}\end{equation}
Finally, we integrate this in the $u$ direction and use the bounds on $\log\bar{\Omega}^2$ on the hypersurfaces $\big\{V=-\underline{v}u\big\}$ and $\big\{u=-1,\ V\in[\underline{v},\underline{v}^{-1}]\big\}$ to get:
\begin{equation}\big|\log\bar{\Omega}^2\big|\leq\kappa^{1-7\delta} 
\label{eq:estimate-on-log-omega-region-II}
\end{equation}
This implies that
\[\big|\bar{\Omega}^2-1\big|\cdot\exp\bigg(D\cdot\frac{V}{u}\bigg)\leq\kappa^{1-8\delta}\]
improving assumption \eqref{A5 II} and proving \eqref{A5 II improved}. This
completes the proof of  \cref{region II bootstrap result}.\qed

As before,  \cref{region II bootstrap result} implies  \cref{region II main result}, giving existence of the solution in region $\mathcal R_{II}$. Moreover, in the $(u,V)$ coordinates, the solution is close to Minkowski space in standard null coordinates. In particular, on the spacelike hypersurface $\big\{u=-\underline{v}V\big\}$ we have the following smallness result:
\begin{corollary}\label{region II corollary}
    We denote $\epsilon=40\delta.$ The solution constructed in   \cref{region II main result} satisfies the following estimates on the spacelike hypersurface $\big\{u=-\underline{v}V\big\}$:
\[\big|\partial_V\phi\big|\leq \kappa^{\frac{1}{2}-\frac{\epsilon}{2}}\cdot V^{-1},\ \big|\partial_u\phi\big|\leq \kappa^{\frac{1}{2}-\frac{\epsilon}{2}}\cdot V^{-1}\]
\[\big|4\partial_Vr-1\big|\leq \kappa^{1-\frac{\epsilon}{2}},\ \big|\partial_ur+1\big|\leq \kappa^{1-\frac{\epsilon}{2}},\ \big|\bar{\Omega}^2-1\big|\leq\kappa^{1-\frac{\epsilon}{2}}\]
\[r=\frac{V}{4}-u+O(V\cdot\kappa^{1-\frac{\epsilon}{2}})\]
\end{corollary}

\subsection{Region \texorpdfstring{$\mathcal R_{III}$}{RIII}}
\label{sec:regRIII}

In this section we prove the existence of the solution in region $\mathcal R_{III}$. Once again, it is convenient to work in the double null coordinates $(u,V)$. The main result of this section is
\begin{proposition}\label{region III main result}
    Let $(M,g,\phi)$ be the solution of \eqref{Einstein Scalar Field} constructed in  \cref{region II main result}. For $\kappa,\underline{v}>0$ sufficiently small, the solution exists in region $\mathcal R_{III}$:
\[ \mathcal R_{III} =\bigg\{(u,V):\ u\in[-1,0),\ V\in(0,1],\ \frac{V}{|u|}\geq\underline{v}^{-1}\bigg\}. \]
Moreover, the solution satisfies the estimates \eqref{A1 III improved}--\eqref{A5 III improved} in region $\mathcal R_{III}$.
\end{proposition}

\subsubsection*{Bootstrap Assumptions}

Let $p>0$ be a small constant (for example $p=1/10$ is sufficient). We consider the following bootstrap assumptions:
\begin{align}\label{A1 III}
    V\cdot\big|\partial_V\phi\big| & \leq 100\kappa^{\frac{1}{2}-\epsilon}
\\ \label{A2 III}
    V\cdot\bigg|\frac{u}{V}\bigg|^p\cdot\big|\partial_u\phi\big| & \leq 100\kappa^{\frac{1}{2}-\epsilon}
\\ \label{A3 III}
    \bigg|\partial_Vr-\frac{1}{4}\bigg| & \leq 100\kappa^{1-2\epsilon}
\\ \label{A4 III}
    \big|\bar{\Omega}^{-2}\partial_ur+1\big|& \leq 100\kappa^{1-2\epsilon}
\\ \label{A5 III}
    \big|\log\bar{\Omega}^2\big| & \leq10\kappa^{1-2\epsilon}\cdot\bigg(1+100\log\bigg|\frac{V\underline{v}}{u}\bigg|\bigg)
\end{align}

A standard local existence argument implies that   \cref{region III main result} follows if we prove the following result:

\begin{proposition}\label{region III bootstrap result}
    For any $(u_1,V_1)\in \mathcal R_{III},$ we assume that the solution satisfies the bootstrap assumptions \eqref{A1 III}--\eqref{A5 III} for all $(u,V)\in\mathcal{R}^{III}_{u_1,V_1}=\mathcal R_{III}\cap\big([-1,u_1]\times[0,V_1]\big).$ Then the solution satisfies the improved estimates in $\mathcal{R}^{III}_{u_1,V_1}$:
   {\allowdisplaybreaks \begin{align}\label{A1 III improved}
    V\cdot\big|\partial_V\phi\big| & \leq 10\kappa^{\frac{1}{2}-\epsilon}
\\ \label{A2 III improved}
    V\cdot\bigg|\frac{u}{V}\bigg|^p\cdot\big|\partial_u\phi\big| & \leq 10\kappa^{\frac{1}{2}-\epsilon}
\\ \label{A3 III improved}
    \bigg|\partial_Vr-\frac{1}{4}\bigg|& \leq 10\kappa^{1-2\epsilon}
\\ \label{A4 III improved}
    \big|\bar{\Omega}^{-2}\partial_ur+1\big|& \leq 10\kappa^{1-2\epsilon}
\\ \label{A5 III improved}
    \big|\log\bar{\Omega}^2\big|&\leq\kappa^{1-2\epsilon}\cdot\bigg(1+10\log\bigg|\frac{V\underline{v}}{u}\bigg|\bigg).
\end{align}}
\end{proposition}

\begin{proof}[Proof of \cref{region III bootstrap result}]
We notice that the bootstrap assumptions hold initially on the spacelike hypersurface $\big\{u=-\underline{v}V\big\}$ because of  \cref{region II corollary}. For the rest of the section, we assume that the bootstrap assumptions hold in $\mathcal{R}^{III}_{u_1,V_1}$ and we prove the desired improved estimates.

\subsubsection*{Preliminary Estimates}

As a consequence of the bootstrap assumptions, we get the bounds:

\[\frac{1}{2}\bigg|\frac{V\underline{v}}{u}\bigg|^{-\frac{p}{2}}\leq\bar{\Omega}^2,\bar{\Omega}^{-2}\leq2\bigg|\frac{V\underline{v}}{u}\bigg|^{\frac{p}{2}},\ \frac{1}{8}\leq|\partial_Vr|\leq\frac{1}{2},\ \big|\partial_ur\big|\leq2\bigg|\frac{V}{u}\bigg|^{\frac{p}{2}}\]
\begin{lemma}
    The area radius function satisfies the following:
    \[r(u,V)=\frac{V}{4}-u+O(\kappa^{1-2\epsilon}V)\]
    In particular, we also have $V/8\leq r\leq V/2.$
\end{lemma}
\begin{proof} 
    The proof follows from \eqref{A3 III} and the estimate for $r$ on the spacelike hypersurface $\{u=-\underline{v}V\}.$
\end{proof}

\subsubsection*{Main Estimates}

We begin by improving assumption \eqref{A1 III}. Using the bootstrap assumptions, the preliminary estimates, and the estimates on the spacelike hypersurface $\big\{u=-\underline{v}V\big\},$ we have from \eqref{wave}:
\begin{align*}\big|r\partial_V\phi\big|(u,V)&\leq\big|r\partial_V\phi\big|(-\underline{v}V,V)+\int_{-\underline{v}V}^u\big|\partial_Vr\cdot\partial_u\phi\big|(u',V)du'\\
& \leq\frac{1}{2}\kappa^{\frac{1}{2}-\epsilon}+O\bigg(\kappa^{\frac{1}{2}-\epsilon}\int_{-\underline{v}V}^uV^{-1}\cdot\bigg|\frac{u'}{V}\bigg|^{-p}du'\bigg).\end{align*}
As a result, we have that
\[\big|r\partial_V\phi\big|(u,V)\leq\frac{1}{2}\kappa^{\frac{1}{2}-\epsilon}+O\big(\kappa^{\frac{1}{2}-\epsilon}\underline{v}^{1-p}\big)\leq\kappa^{\frac{1}{2}-\epsilon}.\]
Since $V\leq8r,$ we improved the assumption \eqref{A1 III} and proved \eqref{A1 III improved}.

Next, we improve the assumption \eqref{A2 III}. We can rewrite the wave equation as:
\[\partial_V\bigg(\bigg|\frac{u}{V}\bigg|^pr\partial_u\phi\bigg)=-\frac{p}{V}\cdot\bigg|\frac{u}{V}\bigg|^pr\partial_u\phi-\bigg|\frac{u}{V}\bigg|^p\partial_ur\partial_V\phi\]
We multiply the equation in order to obtain:
\[\partial_V\bigg(\bigg|\frac{u}{V}\bigg|^pr\partial_u\phi\bigg)^2+\frac{2p}{V}\cdot\bigg(\bigg|\frac{u}{V}\bigg|^pr\partial_u\phi\bigg)^2=-2\bigg|\frac{u}{V}\bigg|^{2p}\partial_ur\partial_V\phi\cdot r\partial_u\phi\]
We integrate this equation and absorb an error term to obtain the energy estimate:
\begin{multline*}
\bigg(\bigg|\frac{u}{V}\bigg|^pr\partial_u\phi\bigg)^2(u,V)+p\int_{-\frac{u}{\underline{v}}}^V\frac{1}{V'}\cdot\bigg(\bigg|\frac{u}{V'}\bigg|^pr\partial_u\phi\bigg)^2dV' \\ \leq\big|\underline{v}^pr\partial_u\phi\big|^2\bigg(u,\frac{-u}{\underline{v}}\bigg)+p^{-1}\int_{-\frac{u}{\underline{v}}}^VV'\cdot\bigg(\bigg|\frac{u}{V'}\bigg|^p\partial_ur\partial_V\phi\bigg)^2dV'.\end{multline*}
Using the bootstrap assumptions, the preliminary estimates, and the bounds of the solution on $\big\{u=-\underline{v}V\big\}$, we get
\[\bigg(\bigg|\frac{u}{V}\bigg|^pr\partial_u\phi\bigg)^2\leq\frac{1}{2}\kappa^{1-2\epsilon}+O\bigg(\kappa^{1-2\epsilon}p^{-1}\int_{-\frac{u}{\underline{v}}}^V(V')^{-1}\cdot\bigg|\frac{u}{V'}\bigg|^pdV'\bigg)\leq\frac{1}{2}\kappa^{1-2\epsilon}+O\bigg(\kappa^{1-2\epsilon}\underline{v}^p\bigg)\leq\kappa^{1-2\epsilon}.\]
Since $V\leq8r,$ we improved the assumption \cref{A2 III} by using the smallness of $\underline{v}$. 

In the following, we will improve assumption \eqref{A3 III}. We can rewrite the wave equation \eqref{partial uv r} for $r$ as:
\[\partial_u\bigg(r\bigg(\partial_Vr-\frac{1}{4}\bigg)\bigg)=-\frac{\bar{\Omega}^2}{4}\cdot\big(\bar{\Omega}^{-2}\partial_ur+1\big)\]
We integrate this and use the bootstrap assumptions, the preliminary estimates, and the bounds of the solution on $\{u=-\underline{v}V\}$ in order to get:
\[r\bigg|\partial_Vr-\frac{1}{4}\bigg|\leq \frac{1}{2}V\kappa^{1-\epsilon}+O\bigg(\kappa^{1-2\epsilon}\int_{-\underline{v}V}^u\bigg|\frac{V}{u'}\bigg|^{\frac{p}{2}}du'\bigg)\leq \frac{1}{2}V\kappa^{1-\epsilon}+O\big(V\kappa^{1-2\epsilon}\underline{v}^{1-\frac{p}{2}}\big)\leq V\kappa^{1-2\epsilon}\]
Using $v\leq8r$ again, we improved the assumption \eqref{A3 III} and proved \eqref{A3 III improved}.

We improve the assumption \eqref{A4 III}. We have from the Raychaudhuri equation \eqref{partial u}:
\[\big|\bar{\Omega}^{-2}\partial_ur+1\big|(u,V)\leq\big|\bar{\Omega}^{-2}\partial_ur+1\big|(-\underline{v}V,V)+\int_{-\underline{v}V}^ur\bar{\Omega}^{-2}(\partial_u\phi)^2(u',V)du'\]
Using the bootstrap assumptions and the preliminary estimates, the second term can be bounded as follows:
\begin{align}\int_{-\underline{v}V}^ur\bar{\Omega}^{-2}(\partial_u\phi)^2(u',V)du' & \lesssim V^{-1}\int_{-\underline{v}V}^u\bigg(V\bigg|\frac{u'}{V}\bigg|^p\big|\partial_u\phi\big|\bigg)^2\cdot\bar{\Omega}^{-2}\cdot\bigg|\frac{u'}{V}\bigg|^{-2p}du' \\ &\lesssim\kappa^{1-2\epsilon}V^{-1}\int_{-\underline{v}V}^u\bigg|\frac{u'}{V}\bigg|^{-\frac{5}{2}p} \underline v^{\frac p2} du'\lesssim\kappa^{1-2\epsilon}\underline{v}^{1-2 p}\end{align}
since $2p<1.$ Using this, the estimates on the spacelike hypersurface $\{u=-\underline{v}V\}$, and the smallness of $\underline{v}$, we get 
\[\big|\bar{\Omega}^{-2}\partial_ur+1\big|(u,V) \leq10\kappa^{1-2\epsilon}.\]
 Thus, we improved the assumption \eqref{A4 III} and proved \eqref{A4 III improved}.

Finally, we improve the assumption \eqref{A5 III}. As in region $\mathcal R_{II}$, an essential step is obtaining estimates on the Hawking mass. We have from \eqref{dv m}:
\[m(u,V)\leq m\bigg(u,-\frac{u}{\underline{v}}\bigg)+\int_{-\frac{u}{\underline{v}}}^V\big(2\partial_Vr\big)^{-1}\cdot\bigg(1-\frac{2m}{r}\bigg)\cdot r^2(\partial_V\phi)^2dV'\]
Using the bootstrap assumptions, the preliminary estimates, and the estimate for $m$ in region $\mathcal R_{II}$, we get:
\[m(u,V)\leq \kappa^{1-\epsilon}\bigg|\frac{u}{\underline{v}}\bigg|+O\bigg(\int_{-\frac{u}{\underline{v}}}^V\kappa^{1-2\epsilon}dV'\bigg)\lesssim V\kappa^{1-2\epsilon}\]
As a result, we obtain that:
\[\frac{m}{r}(u,V)=O\big(\kappa^{1-2\epsilon}\big)\]
We use this estimate in \eqref{du dv Omega using m}, and we get:
\[\partial_u\partial_V\log\bar{\Omega}^2=O\bigg(\kappa^{1-2\epsilon}V^{-2}\bigg|\frac{u}{V}\bigg|^{-p}\bigg)\]
We integrate this and use the estimate \eqref{eq:estimate-partialulogomega-region-II} for $\partial_u\log\bar{\Omega}^2$ in region $\mathcal R_{II}$:
\begin{equation*}\partial_u\log\bar{\Omega}^2(u,v)=\partial_u\log\bar{\Omega}^2\bigg(u,-\frac{u}{\underline{v}}\bigg)+O\bigg(\int_{-\frac{u}{\underline{v}}}^V\kappa^{1-2\epsilon}|u|^{-p}(V')^{-2+p}dV'\bigg)=O\big(\kappa^{1-2\epsilon}\underline{v}^{1-p}|u|^{-1}\big).\end{equation*}
Once again, we integrate this and use the estimate \eqref{eq:estimate-on-log-omega-region-II} for $\log\bar{\Omega}^2$ in region $\mathcal R_{II}$:
\[\big|\log\bar{\Omega}^2\big|(u,V)\leq\big|\log\bar{\Omega}^2\big|(-\underline{v}V,V)+O\bigg(\int_{-\underline{v}V}^u\kappa^{1-2\epsilon}\underline{v}^{1-p}|u'|^{-1}du'\bigg)\leq\kappa^{1-2\epsilon}\bigg(1+\log\bigg|\frac{V\underline{v}}{u}\bigg|\bigg),\]
where we also used the smallness of $\underline{v}.$ In conclusion, we also improved the assumption \eqref{A5 III}, completing the proof of \cref{region III bootstrap result}.
\end{proof}

\subsection{Region \texorpdfstring{$\mathcal R_{IV}$}{RIV}}

In this section we prove existence of the solution in region $\mathcal R_{IV}$ and use the obtained estimates to complete the proof of  \cref{main theorem}, by showing that the constructed spacetime has an incomplete $\mathcal{I}^+.$ In the double null coordinates $(u,V)$, we have that region $\mathcal R_{IV}$ is given by:
\[\mathcal R_{IV}=\bigg\{(u,V):\ u\in[-1,0),\ V\in[1,\infty),\ \frac{V}{|u|}\geq\underline{v}^{-1}\bigg\}\]
To simplify our analysis, we split region $\mathcal R_{IV}$ into the following two sub-regions:
\[\mathcal R_{IV}=\big(\mathcal R_{IV}\cap\{V\leq\underline{v}^{-1}\}\big)\cup[-1,0)\times[\underline{v}^{-1},\infty)\]
We first prove existence of the solution in the first sub-region:
\begin{proposition}\label{region IV part 1 result}
    We denote $\eta=10\epsilon.$ The solution constructed in  \cref{region III main result} extends to region $\mathcal R_{IV}$ for $V\leq\underline{v}^{-1}$, and satisfies the following estimates:
    \begin{align}\label{estimate 1 IV}
        |m|& \leq\kappa^{1-\epsilon}\\
  \label{estimate 2 IV}
        |\partial_ur|& \leq10|u|^{-\frac{p}{2}}
  \\ \label{estimate 3 IV}
        \big|4\partial_Vr-1\big|& \leq\kappa^{1-\eta}
\\ \label{estimate 4 IV}
        \big|\log\bar{\Omega}^2\big|& \leq\kappa^{1-3\eta}\big(1-\log|u|\big)
\\ \label{estimate 5 IV}
        |r\phi|& \leq\kappa^{\frac{1}{2}-\eta}
\\ \label{estimate 6 IV}
        \big|r\partial_V\phi\big|& \leq\kappa^{\frac{1}{2}-\eta}
\\ \label{estimate 7 IV}
        \big|r\partial_u\phi\big| &\leq\kappa^{\frac{1}{2}-\eta}|u|^{-p}
    \end{align}
\end{proposition}
\begin{proof}
    For any $(u_1,V_1)\in \mathcal R_{IV}\cap\{V\leq\underline{v}^{-1}\},$ we assume that the solution satisfies the following bootstrap assumptions for all $(u,V)\in\mathcal{R}^{IV}_{u_1,V_1}=\big([-1,u_1]\times[0,V_1]\big)\cap \mathcal R_{IV}\cap\{V\leq\underline{v}^{-1}\}:$
    \[|\phi|\leq1,\ |\partial_ur|\leq100|u|^{-\frac{p}{2}},\ \frac{1}{8}\leq\partial_Vr\leq\frac{1}{2},\ \bar{\Omega}^2,\bar{\Omega}^{-2}\leq100|u|^{-\frac{p}{2}}\]
    We remark that the bootstrap assumptions hold initially, since the solution satisfies the following estimates on $\big\{u\in[-\underline{v},0),V=1\big\}\cup\big\{V\in[1,\underline{v}^{-1}],\ u=-\underline{v}V\big\}:$
    \[|\partial_ur|\leq4|u|^{-\frac{p}{10}},\ \big|4\partial_Vr-1\big|\leq\kappa^{1-\frac{\eta}{2}},\ \big|\log\bar{\Omega}^2\big|\leq\kappa^{1-\frac{\eta}{2}}\big(1-\log|u|\big)\]
    \[|r\phi|\leq\kappa^{\frac{1}{2}-\frac{\eta}{2}},\ \big|r\partial_V\phi\big|\leq\kappa^{\frac{1}{2}-\frac{\eta}{2}},\ \big|r\partial_u\phi\big|\leq\kappa^{\frac{1}{2}-\frac{\eta}{2}}|u|^{-p}\]
 In the above, the bound for $r\phi$ follows since $\phi(p)=0$ and using our previous bounds for $\partial_u\phi,\ \partial_V\phi.$ To complete the proof, it suffices to prove that the estimates \eqref{estimate 1 IV}--\eqref{estimate 7 IV} hold in $\mathcal{R}^{IV}_{u_1,V_1}.$ Using the bootstrap assumption on $\partial_Vr$, we obtain that in $\mathcal{R}^{IV}_{u_1,V_1}$ the radius function satisfies $r\sim V.$ We first remark that the monotonicity properties of the Hawking mass and the bound on $C_0^+$ imply \eqref{estimate 1 IV}. We use this in the equation:
    \[\partial_u\partial_Vr=-\frac{\bar{\Omega}^2\cdot m}{2r^2}\]
    in order to obtain the estimate:
    \[\partial_Vr=\frac{1}{4}+O(\kappa^{1-\frac{\eta}{2}})+O\bigg(\int_{-\underline{v}V}^u\kappa^{1-\epsilon}V^{-2}|u'|^{-\frac{p}{2}}du'\bigg)=\frac{1}{4}+O(\kappa^{1-\frac{\eta}{2}})\]
    which proves \eqref{estimate 3 IV}. Similarly, we prove \eqref{estimate 2 IV} by the estimate:
    \[|\partial_ur|\leq4|u|^{-\frac{p}{10}}+O\bigg(\int_{\max(1,-\underline{v}u)}^V\kappa^{1-\epsilon}(V')^{-2}|u|^{-\frac{p}{2}}dV'\bigg)\leq10|u|^{-\frac{p}{2}}\]
    Next, we rewrite the wave equation \eqref{wave} as:
    \[\partial_u\partial_V(r\phi)=-\frac{\bar{\Omega}^2\cdot m\cdot\phi}{2r^2}\]
    As a result, we obtain the bound
    \[
    \big|\partial_u(r\phi)\big|\leq C\kappa^{\frac{1}{2}-\frac{\eta}{2}}|u|^{-p}+O\bigg(\int_{\max(1,-\underline{v}u)}^V\kappa^{1-\epsilon}(V')^{-2}|u|^{-\frac{p}{2}}dV'\bigg)\leq2C\kappa^{\frac{1}{2}-\frac{\eta}{2}}|u|^{-p}.
    \]
    We integrate this bound in the $u$ direction to prove \eqref{estimate 5 IV}. Moreover, we also have
    \[\big|r\partial_u\phi\big|\leq2C\kappa^{\frac{1}{2}-\frac{\eta}{2}}|u|^{-p}+\big|\phi\partial_ur\big|\leq\kappa^{\frac{1}{2}-\eta}|u|^{-p}\]
    which proves \eqref{estimate 7 IV}. Similarly, we prove \eqref{estimate 6 IV}:
    \[
    \big|r\partial_V\phi\big|\leq C\kappa^{\frac{1}{2}-\frac{\eta}{2}}+|\phi\partial_Vr|+O\bigg(\int_{-\underline{v}V}^u\kappa^{1-\epsilon}V^{-2}|u'|^{-\frac{p}{2}}du'\bigg)\leq\kappa^{\frac{1}{2}-\eta}.
    \]
    Using equation \eqref{du dv Omega using m} and the previously established bounds, we obtain the estimate
    \[
\big|\partial_V\log\bar{\Omega}^2\big|\leq\kappa^{1-\frac{\eta}{2}}+O\bigg(\int_{-\underline{v}V}^u\kappa^{1-2\eta}V^{-2}|u'|^{-p}du'\bigg)\leq\kappa^{1-\frac{5}{2}\eta}.
    \]
    Finally, we integrate this bound using the estimate on the initial data, in order to obtain \eqref{estimate 4 IV}.
\end{proof}

We now prove existence of the solution in the second sub-region:
\begin{proposition}\label{region IV part 2 result}
    The solution constructed in \cref{region IV part 1 result} extends to $\big\{u\in[-1,0),\ V\in[\underline{v},\infty)\big\}$, and satisfies the following estimates:
    \begin{align}\label{estimate 1 IV part 2}
        |m|& \leq\kappa^{1-\epsilon}
 \\ \label{estimate 2 IV part 2}
        |\partial_ur|& \leq20|u|^{-\frac{p}{2}}
\\ \label{estimate 3 IV part 2}
        \big|4V^{\kappa}\partial_Vr-1\big|& \leq\kappa^{1-2\eta}
  \\ \label{estimate 4 IV part 2}
        \big|\log V^{\kappa}\bar{\Omega}^2\big|& \leq\kappa^{1-5\eta}\big(1-\log|u|\big)
 \\ \label{estimate 5 IV part 2}
        |\phi|& \leq\kappa^{\frac{1}{2}-2\eta}
    \\ \label{estimate 6 IV part 2}
\big|rV^{\kappa}\partial_V\phi\big|& \leq\kappa^{\frac{1}{2}-2\eta}
        \\ \label{estimate 7 IV part 2}
        \big|r\partial_u\phi\big|& \leq\kappa^{\frac{1}{2}-2\eta}|u|^{-p}
    \end{align}
\end{proposition}
\begin{proof}
    For any $(u_1,V_1)\in [-1,0)\times[\underline{v},\infty),$ we assume that the solution satisfies the following bootstrap assumptions for all $(u,V)\in[-1,u_1]\times[\underline{v},V_1]$:
    \begin{equation*}|\phi|\leq1,\ |\partial_ur|\leq100|u|^{-\frac{p}{2}},\ \frac{1}{8}\leq V^{\kappa}\partial_Vr\leq\frac{1}{2},\ V^{\kappa}\bar{\Omega}^2\leq100|u|^{-\frac{p}{2}}.\end{equation*}
    The bootstrap assumptions hold initially because of the estimates in \cref{region IV part 1 result} and the bounds on the initial data on $C_0^+.$ To complete the proof, we prove that the estimates \eqref{estimate 1 IV part 2}--\eqref{estimate 7 IV part 2} hold in $[-1,u_1]\times[\underline{v},V_1].$ Using the bootstrap assumption  on $\partial_Vr$, we obtain that in $[-1,u_1]\times[\underline{v},V_1]$ the radius function satisfies $r\sim V^{1-\kappa}.$ 
    As before, \eqref{estimate 1 IV part 2} follows from the monotonicity properties of the Hawking mass and the bound on $C_0^+$. We use this bound in the equation
    \begin{equation*}\partial_u\big(V^{\kappa}\partial_Vr\big)=-\frac{\bar{\Omega}^2\cdot m\cdot V^{\kappa}}{2r^2}\end{equation*}
    in order to obtain the estimate
\begin{equation*}V^{\kappa}\partial_Vr=\frac{1}{4}+O(\kappa^{1-\eta})+O\bigg(\int_{-1}^u\kappa^{1-\epsilon}V^{-2+2\kappa}|u'|^{-\frac{p}{2}}du'\bigg)=\frac{1}{4}+O(\kappa^{1-\eta})\end{equation*}
    which proves \eqref{estimate 3 IV part 2}. 
    Similarly, we prove \eqref{estimate 2 IV part 2} using
    \begin{equation*}|\partial_ur|\leq10|u|^{-\frac{p}{2}}+O\bigg(\int_{\underline{v}}^V\kappa^{1-\epsilon}(V')^{-2+\kappa}|u|^{-\frac{p}{2}}dV'\bigg)\leq20|u|^{-\frac{p}{2}}.\end{equation*}
    Next, we use \eqref{wave} to obtain the bound
    \begin{equation*}\big|\partial_u(r\phi)\big|\leq C\kappa^{\frac{1}{2}-\eta}|u|^{-p}+O\bigg(\int_{\underline{v}}^V\kappa^{1-\epsilon}(V')^{-2+\kappa}|u|^{-\frac{p}{2}}dV'\bigg)\leq2C\kappa^{\frac{1}{2}-\eta}|u|^{-p}
    \end{equation*}
    We integrate in the $u$ direction and divide by $r$ in order to prove \eqref{estimate 5 IV part 2}. Moreover, we also have
    \begin{equation*}\big|r\partial_u\phi\big|\leq2C\kappa^{\frac{1}{2}-\eta}|u|^{-p}+\big|\phi\partial_ur\big|\leq\kappa^{\frac{1}{2}-2\eta}|u|^{-p}\end{equation*}
    which proves \eqref{estimate 7 IV part 2}. Similarly, we prove \eqref{estimate 6 IV part 2}:
    \begin{equation*}
\big|rV^{\kappa}\partial_V\phi\big|\leq C\kappa^{\frac{1}{2}-\frac{\eta}{2}}+|\phi V^{\kappa}\partial_Vr|+O\bigg(\int_{-1}^u\kappa^{1-\epsilon}V^{-2+2\kappa}|u'|^{-\frac{p}{2}}du'\bigg)\leq\kappa^{\frac{1}{2}-2\eta}  .
 \end{equation*}
    Using equation \eqref{du dv Omega using m} and the previously established bounds, we obtain 
\begin{equation*}   
	 \big|\partial_V\log V^{\kappa}\bar{\Omega}^2\big|=O\bigg(\int_{-1}^u\kappa^{1-4\eta}V^{-2+\kappa}|u'|^{-p}du'\bigg)=O\big(\kappa^{1-4\eta}V^{-2+\kappa}\big).
\end{equation*}
    Since we obtain enough decay in $V,$ we integrate this bound and use \eqref{estimate 4 IV} to obtain \eqref{estimate 4 IV part 2}.
\end{proof}

\subsection{Proof of \texorpdfstring{\cref{main theorem}}{Theorem~1}}
Finally, the estimates that we proved so far allow us to complete the proof of   \cref{main theorem}:
\begin{proof}[Proof of  \cref{main theorem}]
    By construction, the solution satisfies the discrete self-similarity conditions \eqref{self similar properties thm} on $C_0^-$. Our analysis of the solution in regions $\mathcal R_{I}, \mathcal R_{II}, \mathcal R_{III},$ and $\mathcal R_{IV}$ proves that the solution exists in the exterior region $\Rext$. Moreover, the affine length of $\mathcal{I}^+$ is given (see \cite{D05}) by
    \begin{equation*}
\lim_{v\rightarrow\infty}\int_{-1}^0\frac{\Omega^2(u,v)}{\Omega^2(-1,v)}du=\lim_{V\rightarrow\infty}\int_{-1}^0\frac{\bar{\Omega}^2(u,V)}{\bar{\Omega}^2(-1,V)}du=\lim_{V\rightarrow\infty}\int_{-1}^0V^{\kappa}\bar{\Omega}^2(u,V)du.
\end{equation*}
    We use the bound \eqref{estimate 4 IV part 2} in order to get:    \begin{equation*}\lim_{V\rightarrow\infty}\int_{-1}^0V^{\kappa}\bar{\Omega}^2(u,V)du\leq\lim_{V\rightarrow\infty}\int_{-1}^0e^{\kappa^{1-5\eta}}|u|^{-p}du\lesssim 1.\end{equation*}
    In conclusion, the spacetime $\big(M,g\big)$ has a future incomplete $\mathcal{I}^+$.
\end{proof}

\printbibliography[heading=bibintoc] 

\end{document}